\documentclass[10pt,oneside]{article}

\usepackage[T1]{fontenc}
\usepackage[utf8]{inputenc} % (safe even with modern engines)
\usepackage{lmodern}
\usepackage{microtype}

\usepackage[a4paper,margin=2.0cm]{geometry}
\usepackage{setspace}
\usepackage{xcolor}
\definecolor{linkcol}{RGB}{15,70,140}
\definecolor{citecol}{RGB}{0,110,90}
\definecolor{urlcol}{RGB}{120,60,10}

\usepackage{graphicx}
\usepackage{tikz}
\usetikzlibrary{arrows.meta,calc,positioning,decorations.pathmorphing}

\usepackage{amsmath,amssymb,amsfonts,amsthm}
\usepackage{mathtools}
\usepackage{bm}          % bold math symbols
\usepackage{mathrsfs}    % \mathscr
\usepackage{dsfont}      % \mathds
\usepackage{bbm}         % \mathbbm

\usepackage{enumitem}
\setlist{topsep=4pt,itemsep=2pt,parsep=0pt}
\usepackage{booktabs}
\usepackage{array}
\usepackage{pgfplots, bigints}
\pgfplotsset{compat=1.18}
\usepackage{fancyhdr}
\usepackage{titlesec}
\titleformat{\section}{\large\bfseries}{\thesection.}{0.6em}{}
\titleformat{\subsection}{\normalsize\bfseries}{\thesubsection.}{0.6em}{}
\titleformat{\subsubsection}{\normalsize\itshape}{\thesubsubsection.}{0.6em}{}

\numberwithin{equation}{section}

\usepackage{aliascnt}

\theoremstyle{plain}

\newtheorem*{theorem*}{Theorem}
\newaliascnt{proposition}{theorem}
\newtheorem{proposition}[proposition]{Proposition}
\aliascntresetthe{proposition}

\newaliascnt{lemma}{theorem}
\newtheorem{lemma}[lemma]{Lemma}
\aliascntresetthe{lemma}

\newaliascnt{corollary}{theorem}

\aliascntresetthe{corollary}

\theoremstyle{definition}
\newaliascnt{definition}{theorem}
\newtheorem{definition}[definition]{Definition}
\aliascntresetthe{definition}

\newaliascnt{assumption}{theorem}

\aliascntresetthe{assumption}

\newaliascnt{example}{theorem}
\newtheorem{example}[example]{Example}
\aliascntresetthe{example}

\theoremstyle{plain}
\newaliascnt{remark}{theorem}
\newtheorem{remark}[remark]{Remark}
\aliascntresetthe{remark}

\newaliascnt{notation}{theorem}
\newtheorem{notation}[notation]{Notation}
\aliascntresetthe{notation}

\usepackage[
  colorlinks=true,
  linkcolor=linkcol,
  citecolor=citecol,
  urlcolor=urlcol,
  pdfauthor={},
  pdftitle={},
  pdfsubject={},
  pdfkeywords={},
  unicode=true
]{hyperref}

\usepackage[nameinlink,noabbrev]{cleveref}

\usepackage{bookmark} % optional, improves bookmarks
\hypersetup{unicode=true,hypertexnames=false}
\crefname{theorem}{Theorem}{Theorems}
\Crefname{theorem}{Theorem}{Theorems}
\crefname{proposition}{Proposition}{Propositions}
\Crefname{proposition}{Proposition}{Propositions}
\crefname{lemma}{Lemma}{Lemmas}
\Crefname{lemma}{Lemma}{Lemmas}
\crefname{corollary}{Corollary}{Corollaries}
\Crefname{corollary}{Corollary}{Corollaries}
\crefname{definition}{Definition}{Definitions}
\Crefname{definition}{Definition}{Definitions}
\crefname{assumption}{Assumption}{Assumptions}
\Crefname{assumption}{Assumption}{Assumptions}
\crefname{example}{Example}{Examples}
\Crefname{example}{Example}{Examples}
\crefname{remark}{Remark}{Remarks}
\Crefname{remark}{Remark}{Remarks}
\crefname{notation}{Notation}{Notations}
\Crefname{notation}{Notation}{Notations}

\DeclareMathOperator{\Span}{span}

\DeclareMathOperator{\Lie}{\mathcal{L}} % Lie derivative

\newcommand{\R}{\mathbb{R}}

\usepackage[font=small,labelfont=bf]{caption}
\title{\vspace{-1.0em}\bfseries Canonical separating coordinates in the generalized cubic Hénon-Heiles systems}
\author{
Alessandro Portaluri\thanks{Corresponding author} \and
Nicola Sottocornola
}
\date{}

\begin{document}
\maketitle
\begin{abstract}
We study separation of variables for the three classical integrable cubic H\'enon--Heiles systems from the viewpoint of bi-Hamiltonian geometry. Our explicit constructions cover the generalized Kaup--Kupershmidt and KdV$_5$ systems, including the inverse-square term, and the polynomial Sawada--Kotera system. On appropriate open subsets of the canonical phase space, we construct compatible Poisson deformations \(P_1=\Lie_XP_0\), compute the recursion operators \(N=P_1P_0^{-1}\), and determine their control matrices on the invariant codistributions generated by the first integrals.

For the generalized Kaup--Kupershmidt case we give the deformation vector field, Poisson tensor, recursion operator, separating coordinates and canonical conjugate momenta explicitly. We verify the Darboux--Nijenhuis relations and derive the separated curve and its Abelian time quadratures, with local branches and regularity assumptions specified.

The contribution is an explicit reconstruction of known separation variables in a common Poisson--Nijenhuis framework, rather than a claim of new separation variables or conjugate momenta. The generalized Sawada--Kotera integral is recalled, while its explicit separation construction is restricted here to the polynomial case.
\end{abstract}


 \tableofcontents

%%%%%%%%%%%%%%%%%%%%%%%%%%%%%%%%%%%%%%%%%%%%%%%%%%%%
%%%%%
%%%%%
%%%%%
%%%%%
%%%%%
%%%%%
%%%%%%%%%%%%%%%%%%%%%%%%%%%%%%%%%%%%%%%%%%%%%%%%%%%%

\section{Introduction and main results}
The construction of canonical separating coordinates for integrable Hamiltonian systems lies at the intersection of symplectic geometry, Hamilton--Jacobi theory, and the theory of compatible Poisson structures. In many classical examples, separation is first detected through ad hoc changes of variables, Lax representations, algebraic--geometric techniques, or Painlev\'e analysis. The viewpoint adopted in this paper is instead that separation can be extracted directly from the intrinsic geometric data of the system, namely from a compatible Poisson deformation of the canonical tensor and from the associated recursion operator. The bi-Hamiltonian approach originates in Magri's work \cite{Mag78,MM84}; for separation of variables and the Kowalewski conditions we refer to \cite{FP03,Mag21}. B{\l}aszak's monograph \cite{Bla98}, in particular its chapter on finite-dimensional systems, also develops this approach and treats the KdV$_5$ case.

The aim of this paper is to develop this approach for the integrable cubic H\'enon--Heiles family on its natural phase domains. More precisely, we consider the three classical generalized integrable cubic cases --- Kaup--Kupershmidt, KdV$_5$, and Sawada--Kotera --- and study them within the common framework of bi-Hamiltonian geometry, control matrices, and Darboux--Nijenhuis coordinates. The explicit tensor and separation construction for Sawada--Kotera is restricted to the polynomial case \(\mu=0\).

Our main purpose is not to claim novelty for the existence of the separation variables themselves, which are classical and have long been known by other methods, but rather to reconstruct them systematically from compatible Poisson tensors and, above all, to compute explicit canonical conjugate momenta in a fixed coordinate convention, thereby obtaining the corresponding separated Hamiltonian systems. In this sense, the paper is meant as a detailed worked example of the concrete implementation of the bi-Hamiltonian method in finite-dimensional integrable mechanics.

For the generalized Kaup--Kupershmidt case we carry out the construction explicitly in our coordinate convention. The canonical transformations in Verhoeven--Musette--Conte \cite{VMC02} already determine conjugate momenta; we do not claim novelty for these variables. Our purpose is to make their relation with the compatible Poisson tensor and control matrix explicit.

A complementary approach uses Lax pairs and spectral curves. Enolski--Fedorov--Hone \cite{EFH15} derive the SK and KK separations from hyperelliptic Prym varieties associated with trigonal spectral curves and their involutions. The canonical relation between the generalized SK and KK systems is discussed by B{\l}aszak--Rauch-Wojciechowski \cite{BRW94}; see also \cite{VMC02}. These results provide an algebraic-geometric interpretation of the separated curves below.

Beyond its intrinsic interest, such an explicit description in canonical separated variables provides a natural starting point for the computation of symplectic and variational indices -- including Maslov-type, Morse, and H\"ormander indices -- associated with distinguished classes of trajectories, such as periodic, brake, homoclinic, and heteroclinic solutions. These indices, in turn, offer a useful tool for probing the fine structure of the phase space and its dynamical behavior.

\paragraph{Integrable generalized cubic H\'enon--Heiles systems.}
We assume \(a\ne0\) throughout the nondegenerate cubic family. The generalized cubic H\'enon--Heiles Hamiltonian is usually written in the form
\begin{equation}\label{eq:gHH}
H=
\dfrac12\bigl(p_x^2+p_y^2\bigr)
+\dfrac12\bigl(\omega_1x^2+\omega_2y^2\bigr)
+a\,x y^2-\dfrac{b}{3}\,x^3
+\dfrac{\mu}{2y^2},
\end{equation}
and the Painlev\'e analysis, together with the direct construction of an additional
first integral, shows that \eqref{eq:gHH} is Liouville integrable  in the following three
cases, namely the Sawada--Kotera (SK), fifth-order KdV (KdV$_5$), and
Kaup--Kupershmidt (KK) cases:
\begin{equation}\label{eq:cubicHHcases}
\begin{array}{rcll}
\text{SK:}      & \displaystyle \dfrac{b}{a}=-1,   & \omega_1=\omega_2,      & \text{momentum degree }4\ (\mu\ne0),\ 2\ (\mu=0),\\[3ex]
\text{KdV$_5$:} & \displaystyle \dfrac{b}{a}=-6,   & \omega_1,\omega_2\ \text{arbitrary}, & \text{momentum degree }2,\\[3ex]
\text{KK:}      & \displaystyle \dfrac{b}{a}=-16,  & \omega_1=16\,\omega_2,  & \text{momentum degree }4.
\end{array}
\end{equation}
These three systems are related, respectively, to stationary reductions of the
Sawada--Kotera, fifth-order KdV, and Kaup--Kupershmidt equations; see Fordy \cite{For} and \cite{VMC02}. 

In \eqref{eq:gHH} we consider the inverse-square perturbation only. For \(\mu\ne0\) the natural real phase domain is \(y\ne0\). At \(\mu=0\) the polynomial Hamiltonians extend across \(y=0\), although some of the separation charts below do not. Without an affine translation in \(x\), convenient representatives are the following.

\medskip
\noindent
\textbf{(i) Sawada--Kotera case.}
\begin{equation}\label{eq:SK}
H_{\mathrm{SK}}
=
\dfrac12\bigl(p_x^2+p_y^2\bigr)
+\dfrac{\omega}{2}\bigl(x^2+y^2\bigr)
+a\,x y^2+\dfrac{a}{3}x^3+\dfrac{\mu}{2y^2}.
\end{equation}
The generalized second integral \(K_{\mathrm{SK}}\) is quartic in the momenta; for \(\mu=0\) one may use the unsquared integral \(K_{\mathrm{SK},0}\), quadratic in the momenta.

\medskip
\noindent
\textbf{(ii) Fifth-order KdV case.}
\begin{equation}\label{eq:KdV5}
H_{\mathrm{KdV}_5}
=
\dfrac12\bigl(p_x^2+p_y^2\bigr)
+\dfrac12\bigl(\omega_1x^2+\omega_2y^2\bigr)
+a\,x y^2+2a\,x^3+\dfrac{\mu}{2y^2}.
\end{equation}
This system possesses a second first integral \(K_{\mathrm{KdV}_5}\) which is quadratic
in the momenta.

\medskip
\noindent
\textbf{(iii) Kaup--Kupershmidt case.}
\begin{equation}\label{eq:KK}
H_{\mathrm{KK}}
=
\dfrac12\bigl(p_x^2+p_y^2\bigr)
+8\omega x^2+\dfrac{\omega}{2}y^2
+a\,x y^2+\dfrac{16a}{3}x^3+\dfrac{\mu}{2y^2}.
\end{equation} 
\begin{remark}
The three parameter choices in \eqref{eq:cubicHHcases} are the classical integrable cubic H\'enon--Heiles cases \cite{For,VMC02}. We use their explicit integrals and do not prove a converse classification for arbitrary smooth first integrals.

For two degrees of freedom, Liouville integrability on a phase domain means that \(H,K\) are defined there, satisfy \(\{H,K\}_0=0\), and have \(dH\wedge dK\ne0\) on an open dense subset. Regular common levels are Lagrangian. The particular separation charts studied here are local canonical charts; no global separation in the original configuration coordinates is asserted.
\end{remark}
We use the standard symplectic structure, restricted to the phase domain just specified,
\[
(T^*\R^2,\omega_0),\qquad
\omega_0=dp_x\wedge dx+dp_y\wedge dy,
\]
with canonical coordinates ordered as
\[
(p_x,p_y,x,y).
\]
In these coordinates, the associated canonical Poisson tensor is
\[
P_0=\omega_0^{-1}=
\begin{pmatrix}
0&0&-1&0\\
0&0&0&-1\\
1&0&0&0\\
0&1&0&0
\end{pmatrix}.
\]
\paragraph{Second first integrals for the three generalized cubic H\'enon--Heiles cases.}
For the generalized cubic H\'enon--Heiles family \eqref{eq:gHH}
we consider the three classical cases $\text{SK},\text{KdV}_5,\text{KK}$. A convenient choice of the additional first integrals is given below.  We stress that these integrals are not unique: one may always replace \(K\) by
\(K+\Phi(H)\), where \(\Phi\) is any polynomial function, and obtain an equivalent
Liouville-integrable pair. The formulas below are written in the original,
unshifted variables \((p_x,p_y,x,y)\). The existence of these three integrable
cases, and the standard polynomial representatives of the extra integrals, goes
back to the classical work of Fordy and to the later synthesis of the generalized
cubic family by Verhoeven--Musette--Conte \cite{For,VMC02}. 

\medskip
\noindent
\textbf{(iii) Kaup--Kupershmidt case.}
Case (iii) of the generalized cubic H\'enon--Heiles family is the Kaup--Kupershmidt case. In this situation the Hamiltonian is
\[
H_{\mathrm{KK}}
=
\dfrac12\bigl(p_x^2+p_y^2\bigr)
+8\omega x^2+\dfrac{\omega}{2}y^2
+a\,x y^2+\dfrac{16a}{3}x^3+\dfrac{\mu}{2y^2},
\]
namely a natural Hamiltonian with cubic potential supplemented by the inverse--square term \(\mu/(2y^2)\). Complete integrability is ensured by the existence of a second independent first integral, quartic in the momenta, which can be chosen as
\[
K_{\mathrm{KK}}
=
9\left(p_y^2+\omega y^2+\dfrac{\mu}{y^2}\right)^2
+12a\,y^2p_y\bigl(3x\,p_y-y\,p_x\bigr)
-2a^2y^4\bigl(6x^2+y^2\bigr)
+12ax\bigl(\mu-\omega y^4\bigr)
-18\omega\mu.
\]
Thus the pair \(\bigl(H_{\mathrm{KK}},K_{\mathrm{KK}}\bigr)\) provides the two commuting integrals required by Liouville integrability. In the bi-Hamiltonian approach, this case serves as the guiding model: starting from a suitable deformation vector field, one constructs the deformed Poisson tensor, the associated recursion operator, and the control matrix, from which the separating coordinates and the conjugate momenta are then obtained explicitly.

\medskip
\noindent
\textbf{(ii) Fifth-order KdV case.}
The second integrable member is the KdV$_5$ case, for which
\[
H_{\mathrm{KdV}_5}
=
\dfrac12\bigl(p_x^2+p_y^2\bigr)
+\dfrac12\bigl(\omega_1x^2+\omega_2y^2\bigr)
+a\,x y^2+2a\,x^3+\dfrac{\mu}{2y^2}.
\]
Also in this case, complete integrability is provided by the existence of a second independent first integral. A convenient representative is
\[
K_{\mathrm{KdV}_5}
=
4ay\,p_xp_y
+\bigl(4\omega_2-\omega_1-4ax\bigr)
\left(p_y^2+\dfrac{\mu}{y^2}\right)
+a^2y^4
+4a^2x^2y^2
+4a\omega_2xy^2
+\omega_2(4\omega_2-\omega_1)y^2.
\]
In the polynomial limit \(\mu=0\), this reduces to the classical additional integral of the cubic KdV$_5$ H\'enon--Heiles system. As in the Kaup--Kupershmidt case, the bi-Hamiltonian construction is performed through a suitable deformation vector field: this yields the deformed Poisson tensor and the recursion operator, and from the induced control matrix one extracts the corresponding separating variables.

\medskip
\noindent
\textbf{(i) Sawada--Kotera case.}
Case (i) is the Sawada--Kotera case, whose Hamiltonian is
\[
H_{\mathrm{SK}}
=
\dfrac12\bigl(p_x^2+p_y^2\bigr)
+\dfrac{\omega}{2}(x^2+y^2)
+a\,x y^2+\dfrac{a}{3}x^3+\dfrac{\mu}{2y^2}.
\]
Here one may first introduce the expression of total polynomial degree three
\[
K_{\mathrm{SK},0}
=
3p_xp_y+a\,y^3+3a\,x^2y+3\omega xy,
\]
which in the polynomial case \(\mu=0\) is itself the extra first integral. In the generalized case \(\mu\neq0\), the additional first integral takes the well-known squared form
\[
K_{\mathrm{SK}}
=
K_{\mathrm{SK},0}^{\,2}
+\mu\left(\dfrac{9p_x^2}{y^2}+12ax\right),
\]
that is,
\[
K_{\mathrm{SK}}
=
\left(3p_xp_y+a\,y^3+3a\,x^2y+3\omega xy\right)^2
+\mu\left(\dfrac{9p_x^2}{y^2}+12ax\right).
\]
Below we construct the SK deformation and separation data for \(\mu=0\), using the unsquared integral \(F_{\mathrm{SK}}=K_{\mathrm{SK},0}/3\). For the generalized separation and its relation to KK, see \cite{BRW94,VMC02,EFH15}.
\medskip
\noindent
In each of the three cases above, the Hamiltonian \(H\) and the corresponding
additional integral \(K\) are functionally independent on an open dense subset of phase space. This is
precisely the sense in which these systems are Liouville integrable: for a
Hamiltonian system with two degrees of freedom, the pair \((H,K)\) provides two
independent commuting first integrals, so that the regular common level sets are
invariant Lagrangian manifolds and the dynamics can be integrated by quadratures. 

The existence of separating variables for the (standard) cubic H\'enon--Heiles
family is classical: it goes back at least to the work of
Ravoson, Gavrilov, and Caboz \cite{RGC93}, who exhibited separation variables and
corresponding Lax structures for the integrable cubic cases, and it has
subsequently been revisited by several authors from the point of view of
Poisson geometry, natural Poisson bivectors, and canonical or B\"acklund
transformations. What we emphasize here is a different aspect of the same
story, namely the explicit geometric passage
\[
(H,K)\quad\Longrightarrow\quad X\quad\Longrightarrow\quad
P_1=\Lie_XP_0\quad\Longrightarrow\quad
N=P_1P_0^{-1}\quad\Longrightarrow\quad
M\quad\Longrightarrow\quad
(u_1,u_2)\quad\Longrightarrow\quad
(v_1,v_2),
\]
where \(X\) is a deformation vector field, \(P_1\) is a compatible Poisson
tensor, \(N\) is the recursion operator, \(M\) is the control matrix acting
on the codistribution generated by \(dh\) and \(df\), and
\((v_1,v_2, u_1,u_2)\) are canonical separated variables. This presentation builds on the compatible-tensor, recursion-operator and control-matrix construction in \cite{Sotto}, making the local hypotheses, canonical momenta and subsequent separated dynamics explicit.

The general philosophy is the following. Starting from the canonical Poisson
tensor \(P_0\), one seeks a second Poisson tensor \(P_1\) compatible with
it. In the present paper, we adopt the exact-deformation ansatz
$P_1=\Lie_XP_0$,   where \(X\) is a suitable vector field. Since \(P_0\) is nondegenerate, this leads to the recursion operator $ N=P_1P_0^{-1}$,
and the torsionless character of \(N\) becomes the key structural property. Finding a suitable \(X\) and verifying that \(\Lie_XP_0\) is Poisson and preserves the integral codistribution remain problem-dependent steps; this ansatz is not a general algorithm for finding separation variables.
On an open set where \(N\) is semisimple with two distinct, functionally independent eigenvalues, each of multiplicity two, those eigenvalues are candidates for local separation coordinates. At the same time, the dual operator \(N^*\) acts on the codistribution
\[
\Span\{dH,dK\}\subset T^*M
\]
through a \(2\times2\) control matrix \(M\), and in the cases under
consideration the eigenvalues of \(M\) coincide with the two distinct eigenvalues
of \(N\). In this way, the control matrix furnishes a finite-dimensional
shadow of the full recursion operator and connects the bi-Hamiltonian
construction with the Kowalewski separability conditions.

Once the separating coordinates \(u_1,u_2\) have been found, the next step is
to determine conjugate momenta \(v_1,v_2\) such that
\[
\omega_0=dv_1\wedge du_1+dv_2\wedge du_2.
\]
Although the separating variables themselves are classical, the explicit
construction of canonical conjugate momenta in one fixed convention is often
left implicit in the literature, encoded through generating functions,
separated relations, or canonical transformations written in different
coordinates. One of the practical aims of this paper is therefore to make
this step completely explicit in the coordinate convention
\[
(p_x,p_y,x,y),
\]
and to exhibit the corresponding Hamilton equations in separated canonical
variables.

The three generalized cubic H\'enon--Heiles systems considered here share a common
integrable background, but they display different geometric features. In the
Kaup--Kupershmidt case, the second integral is quartic and the deformation
vector field is nontrivial; the resulting recursion operator has a genuinely
non-obvious spectrum, and the conjugate momenta require a substantial
calculation based on the Darboux--Nijenhuis conditions. For this reason, we
develop the Kaup--Kupershmidt case in full detail and use it as the main
worked example of the paper. The KdV$_5$ case fits the same general scheme,
but with a simpler control matrix and a different separated geometry, leading
to explicit separating coordinates obtained from the spectrum of the
recursion operator. The polynomial Sawada--Kotera case is distinguished by the fact
that, in our convention, the separation ultimately reduces to rotated
Cartesian variables; nevertheless, it also admits a natural description in
terms of a deformation vector field, a compatible Poisson tensor, a
recursion operator, and a control matrix.

The paper is organized as follows. In
\Cref{sec:preliminaries-placeholder} we recall the geometric ingredients
needed in the sequel: compatible Poisson tensors, recursion operators,
control matrices, companion matrices, and Kowalewski conditions on
\(T^*\R^2\). In the central section we turn to the integrable generalized  cubic
H\'enon--Heiles family. The Kaup--Kupershmidt case is treated in full detail,
from the construction of the deformation vector field to the explicit
computation of the separating coordinates and canonical momenta. We then
discuss the KdV$_5$ and Sawada--Kotera cases in the same spirit, with
particular emphasis on the corresponding separating coordinates, control
matrices, and canonical transformations. 
% The appendices collect some
% auxiliary material on exact deformations of the canonical Poisson tensor and
% on the Fr\"olicher--Nijenhuis differential associated with a torsionless
% recursion operator.

We stress once more that the emphasis throughout is constructive. The cubic
H\'enon--Heiles systems have long served as a testing ground for different
methods of integrability and separation of variables. Our purpose is to show
in detail how the same systems may be revisited through compatible Poisson
geometry, so that the separating variables emerge from the spectrum of the
recursion operator and the canonical momenta are reconstructed in an
explicit and uniform way. In this sense, the paper provides a concrete bridge
between classical separation theory and the modern language of
Poisson--Nijenhuis geometry.

%%%%%%%%%%%%%%%%%%%%%%%%%%%%%%%%%%%%%%%%%%%%%%%%%%%%
%%
%%
%%
%%
%%
%%
%%%%%%%%%%%%%%%%%%%%%%%%%%%%%%%%%%%%%%%%%%%%%%%%%%%%

\section{\texorpdfstring{Preliminaries on Bi-Hamiltonian Systems}{Preliminaries on Bi-Hamiltonian Systems}}\label{sec:preliminaries-placeholder}

This section summarizes the geometric ingredients underlying our construction of separation variables on \(T^*\R^2\). We recall the basic definitions of compatible Poisson tensors, recursion operator, and bi-Hamiltonian structure, and then introduce the control matrix and the transition matrix associated with two involutive integrals. These notions lead to a local separation criterion under exact invariance of the integral codistribution; the differential identities called Kowalewski conditions below are recorded separately.

%%%%%%%%%%%%%%%%%%%%%%%%%%%%%%%%%%%%%%%%%%%%%%%%%%%%
%%
%%
%%
%%
%%
%%
%%%%%%%%%%%%%%%%%%%%%%%%%%%%%%%%%%%%%%%%%%%%%%%%%%%%

\subsection{Poisson tensors, compatibility and recursion operator}
Let \(M\) be an open subset of \(T^*\R^2\) and use Darboux coordinates ordered as
\[
(p,q)=(p_x,p_y,x,y),\qquad z=(z^1,z^2,z^3,z^4)=(p_x,p_y,x,y).
\]
The Liouville one--form and canonical symplectic form are
\[
\theta=p_x\,dx+p_y\,dy,\qquad
\omega_0=d\theta=dp_x\wedge dx+dp_y\wedge dy.
\]
We associate to \(\omega_0\) the skew-symmetric matrix \(\Omega_0=(\Omega_0)_{ij}\) by evaluation,
\[
(\Omega_0)_{ij}:=\omega_0\!\left(\dfrac {\partial}{\partial z^i},\dfrac {\partial}{\partial z^j}\right),
\qquad
\omega_0=\dfrac 12\sum_{i,j=1}^4(\Omega_0)_{ij}\,dz^i\wedge dz^j.
\]
We use \(P_0=\Omega_0^{-1}\) as the coordinate matrix of the Poisson tensor. With \(\omega_0^\flat(V)=\omega_0(V,\cdot)\), the bundle-map convention is \(P_0^\sharp=-(\omega_0^\flat)^{-1}\); in particular \(\iota_{P_0dH}\omega_0=-dH\).
Since \(\Omega_0^2=-I\), one has \(\Omega_0^{-1}=-\Omega_0\), hence
\[
P_0=-\Omega_0=
\begin{pmatrix}
0&0&-1&0\\
0&0&0&-1\\
1&0&0&0\\
0&1&0&0
\end{pmatrix}.
\]
The associated Poisson bracket is
\[
\{f,g\}_0=\langle df,\,P_0\,dg\rangle
=
\dfrac {\partial f}{\partial x}\dfrac {\partial g}{\partial p_x}
-\dfrac {\partial f}{\partial p_x}\dfrac {\partial g}{\partial x}
+
\dfrac {\partial f}{\partial y}\dfrac {\partial g}{\partial p_y}
-\dfrac {\partial f}{\partial p_y}\dfrac {\partial g}{\partial y},
\]
so that
\[
\{x,p_x\}_0=\{y,p_y\}_0=1,
\qquad
\{p_x,x\}_0=\{p_y,y\}_0=-1.
\]
We represent covectors by coordinate rows and tangent vectors by columns. Thus the dual action and Poisson bundle map have coordinate expressions
\[
N^*\alpha=\alpha N,\qquad P_0^\sharp(\alpha)=P_0\alpha^\top.
\]
Expressions such as \(N^*dH\) and \(P_0dH\) below denote the intrinsic bundle-map actions; in row coordinates they are \(dH\,N\) and \(P_0(dH)^\top\), respectively. If we write the gradient column as \(\nabla f=(\partial_p f,\partial_q f)^\top\), then
\[
\{f,g\}_0=(\nabla f)^\top P_0\,\nabla g.
\]
Fix the standard Euclidean metric
\[
g=dp_x^2+dp_y^2+dx^2+dy^2
\]
and let \(\sharp_g:T^*M\to TM\) be the induced identification. Define the \((1,1)\)-tensor
\[
J_0:=\sharp_g\circ \omega_0^\flat,
\qquad\text{equivalently}\qquad
\omega_0(X,Y)=g(J_0X,Y).
\]
In the coordinates \(z=(p_x,p_y,x,y)\) the matrix of \(g\) is the identity, hence \(J_0=-\Omega_0\) and \(J_0^2=-I\).
Since \(P_0=\Omega_0^{-1}=-\Omega_0\), we also have the coordinate identity \(P_0=J_0\).

\begin{notation}[Warning on conventions]
Throughout this paper we use the evaluation convention \((\Omega_0)_{ij}=\omega_0(\partial_{z^i},\partial_{z^j})\), hence
\(P_0=\Omega_0^{-1}\) and
\[
\{x,p_x\}_0=\{y,p_y\}_0=1
\]
in the ordered coordinates \((p_x,p_y,x,y)\).
\end{notation}

\begin{definition}
Let \(P,Q\in \mathfrak X^2(T^*\R^2)\), where \(\mathfrak X^2(T^*\R^2)\) denotes the space of smooth contravariant skew-symmetric \(2\)-tensor fields on \(T^*\R^2\). If \(P=X_1\wedge X_2\), then the Schouten bracket of \(P\) and \(Q\) is defined by
\[
[P,Q]=\Lie_{X_1}Q\wedge X_2-\Lie_{X_2}Q\wedge X_1.
\]
In particular, \(P\) is a Poisson bivector if and only if
\[
[P,P]=0.
\]
\end{definition}
\begin{remark}
It's worth observing that since \(P_0\) is constant, one has
\[
P_1=-(J_XP_0+P_0J_X^\top),
\]
where \(J_X\) denotes the Jacobian matrix of \(X\).
\end{remark}
\begin{definition}[Poisson tensors and compatibility]
A bivector field \(P\in\mathfrak X^2(T^*\R^2)\) is a \emph{Poisson tensor} if
\[
[P,P]=0.
\]
It defines the Poisson bracket
\[
\{f,g\}_P=P(df,dg)=\langle df,P^\sharp(dg)\rangle.
\]
Two Poisson tensors \(P_0,P_1\) are said to be \emph{compatible} if
\[
[P_0,P_1]=0.
\]
\end{definition}
Let
\[
P_1=
\begin{pmatrix}
K & L\\
- L^\top & M
\end{pmatrix},
\]
with \(K^\top=-K\) and \(M^\top=-M\). Then
the coordinate matrix of the associated recursion operator and its transpose are
\[
N=P_1P_0^{-1}
=
\begin{pmatrix}
- L & K\\
- M & -L^\top
\end{pmatrix},
\qquad
N^\top=P_0^{-1}P_1
=
\begin{pmatrix}
- L^\top & M\\
- K & -L
\end{pmatrix}.
\]
\begin{definition}[Nijenhuis torsion]
The Nijenhuis torsion of a \((1,1)\)-tensor \(N\) is
\[
T_N(X,Y)=[NX,NY]-N\bigl([NX,Y]+[X,NY]-N[X,Y]\bigr).
\]
If \(T_N=0\), then \(N\) is said to be \emph{torsionless}.
\end{definition}

\begin{lemma}\label{lem:PN}
If \(P_0\) is nondegenerate and \(P_1\) is a Poisson tensor compatible with \(P_0\), then
\[
N=P_1P_0^{-1}
\]
is a Nijenhuis operator, i.e. \(T_N=0\). Hence \((P_0,N)\) is a Poisson--Nijenhuis structure.
\end{lemma}

\begin{proof}
This is the nondegenerate compatible-Poisson-pair theorem of Poisson--Nijenhuis geometry \cite{MM84,FP03}: if two Poisson tensors are compatible and the first is invertible, their quotient \(P_1P_0^{-1}\) has zero Nijenhuis torsion and satisfies \(NP_0=P_0N^*=P_1\). The hypotheses of that theorem are exactly those stated here. Thus \(T_N=0\), and \((P_0,N)\) is Poisson--Nijenhuis. In the applications we also verify the Jacobi, compatibility and torsion identities directly.
\end{proof}

%%%%%%%%%%%%%%%%%%%%%%%%%%%%%%%%%%%%%%%%%%%%%%%%%%%%
%%
%%
%%
%%
%%
%%
%%%%%%%%%%%%%%%%%%%%%%%%%%%%%%%%%%%%%%%%%%%%%%%%%%%%
\subsection{Control matrix}

Let \(H_1,H_2\) be independent involutive integrals on an open set \(U\subset M\), and set
\[
X_a:=P_0\,dH_a,\qquad a=1,2.
\]
Their common level sets define a regular \(2\)-dimensional foliation with tangent distribution
\[
\mathcal D=\ker(dH_1,dH_2)=\mathrm{span}\{X_1,X_2\}.
\]
Accordingly, one may choose local adapted coordinates \((H_1,H_2,u^1,u^2)\) such that
\[
\mathcal D=\mathrm{span}\Bigl\{\dfrac{\partial}{\partial u^1},\dfrac{\partial}{\partial u^2}\Bigr\}.
\]
Thus there exists an invertible matrix \(\Phi=(\Phi_j^{\,a})\) for which
\[
\dfrac{\partial}{\partial u^j}=\sum_{a=1}^2 \Phi_j^{\,a}X_a,
\qquad j=1,2.
\]
In a canonical separation chart, the transition matrix has the St{\"a}ckel row property
\[
\Phi_j^{\,a}=\Phi_j^{\,a}(u^j,H_1,H_2).
\]
The existence of such a chart will follow from the hypotheses of \Cref{prop:Kow-main}, rather than from a choice of arbitrary adapted coordinates.
\begin{definition}[Control matrix]
Let \(H_1,H_2\) be two independent involutive integrals, and set
\[
\mathcal D^*:=\mathrm{span}\{dH_1,dH_2\}\subset T^*M.
\]
Choose a splitting
\[
T^*M=\mathcal D^*\oplus\mathcal K,
\]
and denote by \(\pi_{\mathcal D^*}\) the projection onto \(\mathcal D^*\) along \(\mathcal K\). The \emph{control matrix} associated with \((H_1,H_2)\) is the \(2\times2\) matrix \(M=(M_{ab})\) defined by
\[
\pi_{\mathcal D^*}(N^*dH_a)=\sum_{b=1}^2 M_{ab}\,dH_b,
\qquad a=1,2.
\]
Equivalently,
\[
N^*dH_a=\sum_{b=1}^2 M_{ab}\,dH_b+\eta_a,
\qquad \eta_a\in\mathcal K,
\qquad a=1,2.
\]
\end{definition}

This projected definition is meaningful for a chosen complement. For the spectral and separation arguments below, we additionally require
\[
N^*\mathcal D^*\subseteq\mathcal D^*.
\]
Then \(\eta_a=0\), the projection is redundant, and the control matrix is independent of the complement:
\[
N^*dH_a=\sum_{b=1}^2M_{ab}\,dH_b,\qquad a=1,2.
\]
Without this invariance assumption, the eigenvalues of the projected matrix need not be eigenvalues of \(N\). Under the regularity assumptions of \Cref{prop:Kow-main}, the two eigenvalues of the invariant control matrix are the distinct eigenvalues of \(N\) and provide local separation coordinates.

With covectors represented as rows, put
\[
D_H=\begin{pmatrix}dH_1\\dH_2\end{pmatrix},
\qquad \mathsf X=(X_1\ X_2)=P_0D_H^\top.
\]
The same invariant control identities read
\[
D_HN=MD_H,\qquad N\mathsf X=\mathsf X M^\top.
\]
In particular, \(NX_a=\sum_bM_{ab}X_b\). These formulas preserve the row-covector and column-vector convention; no transposition of the displayed control matrices is required.

%%%%%%%%%%%%%%%%%%%%%%%%%%%%%%%%%%%%%%%%%%%%%%%%%%%%
%%
%%
%%
%%
%%
%%
%%%%%%%%%%%%%%%%%%%%%%%%%%%%%%%%%%%%%%%%%%%%%%%%%%%%

\subsection{Polynomial encoding and Kowalewski condition}

Let
\[
q(\lambda)=\lambda^2-s_1\lambda-s_2
\]
be the monic quadratic polynomial whose coefficients \(s_1,s_2\in C^\infty(U)\) are the spectral invariants of \(M\). On the open set where the discriminant
\[
\Delta=s_1^2+4s_2
\]
does not vanish, the roots \(u^1,u^2\) of \(q\) are simple and satisfy
\[
q(\lambda)=(\lambda-u^1)(\lambda-u^2),
\qquad
s_1=u^1+u^2,
\qquad
s_2=-u^1u^2.
\]
The corresponding Jacobian along the commuting frame is
\[
J_q=
\begin{pmatrix}
X_1(s_1) & X_1(s_2)\\
X_2(s_1) & X_2(s_2)
\end{pmatrix},
\]
and, whenever \(\det J_q\neq 0\), the companion matrix
\[
C_q=
\begin{pmatrix}
s_1&s_2\\
1&0
\end{pmatrix}
\]
induces
\[
\widetilde M_q:=J_q\,C_q\,J_q^{-1}.
\]
If exact invariance holds and \(N^*ds_1=s_1ds_1+ds_2\), \(N^*ds_2=s_2ds_1\), evaluating these identities on \(X_a\) gives \(MJ_q=J_qC_q\). Hence \(\widetilde M_q=M\) wherever \(J_q\) is invertible.

Let
\[
K:\mathcal D\to\mathcal D,
\qquad
K(X_a)=\sum_{b=1}^2 M_{ab}X_b,
\qquad a=1,2.
\]
The Kowalewski condition is the differential constraint
\[
X_1(M_{21})=X_2(M_{11}),
\qquad
X_1(M_{22})=X_2(M_{12}).
\]
These two identities hold for the control matrices in our examples. They are not, for an arbitrary matrix in a commuting frame, equivalent to vanishing Nijenhuis torsion or sufficient by themselves for canonical separation. The local statement used here is the following.

\begin{proposition}\label{prop:Kow-main}
Let \((P_0,N)\) be a symplectic Poisson--Nijenhuis structure in dimension four. Suppose \(N\) is semisimple with two distinct, functionally independent eigenvalues \(u^1,u^2\), each of multiplicity two. Let \(H_1,H_2\) be independent involutive integrals and assume
\[
N^*\operatorname{span}\{dH_1,dH_2\}\subseteq\operatorname{span}\{dH_1,dH_2\},
\qquad du^1\wedge du^2\wedge dH_1\wedge dH_2\ne0.
\]
Then locally there exist conjugate momenta \(v_1,v_2\) such that
\[
\omega_0=\sum_{i=1}^2dv_i\wedge du^i,\qquad
N^*du^i=u^i du^i,\qquad N^*dv_i=u^i dv_i.
\]
The common levels have separated relations
\(v_i=\psi_i(u^i,H_1,H_2)\), with
\(\det(\partial\psi_i/\partial H_a)\ne0\). Thus the Hamilton--Jacobi equation separates locally in these Darboux--Nijenhuis coordinates.
\end{proposition}

\begin{proof}
The local Darboux--Nijenhuis theorem gives the stated canonical chart under the spectral hypotheses; see \cite[Proposition~2.5]{FP03}. The tangent distribution \(\mathcal D\) of the common levels is Lagrangian and \(N\)-invariant. The two symplectic eigenspaces are mutually symplectically orthogonal, and invariance decomposes \(\mathcal D\) into one line in each eigenspace: a two-dimensional intersection with either eigenspace would contradict the Lagrangian property. The transversality assumption makes \((u^1,u^2,H_1,H_2)\) local coordinates. On a fixed level write \(v_i=\psi_i(u^1,u^2,H_1,H_2)\). The two tangent lines lie in their respective canonical planes, so \(\partial\psi_i/\partial u^j=0\) for \(i\ne j\). The coordinate change implies that \((\partial\psi_i/\partial H_a)\) is invertible. The separated generating function is locally the sum of primitives of \(\psi_i(u^i,H_1,H_2)\,du^i\). This is the invariant-Lagrangian separation criterion of \cite[Propositions~3.2--3.3]{FP03}.
\end{proof}

%%%%%%%%%%%%%%%%%%%%%%%%%%%%%%%%%%%%%%%%%%%%%%%%%%%%
%%
%%
%%
%%
%%
%%
%%%%%%%%%%%%%%%%%%%%%%%%%%%%%%%%%%%%%%%%%%%%%%%%%%%%
\paragraph{Local companion formulation.}
The rank--$2$ distribution
\[
\mathcal D:=\mathrm{span}\{X_h,X_f\}\subset TM
\]
is commuting and integrates to a foliation; its leaves are denoted by $F$, so that
\[
TF=\mathcal D|_F.
\]
Let $\mathrm{Ann}(\mathcal D)\subset T^*M$ be the annihilator bundle; on each leaf
\[
\mathrm{Ann}(TF)=\mathrm{span}\{dh,df\}\subset T^*M|_F.
\]

\noindent
Assume there exist smooth functions $s_1,s_2\in C^\infty(M)$ such that
\[
ds_1\wedge ds_2\wedge dh\wedge df\neq 0
\quad\text{on an open set }U\subset M,
\]
so that $(s_1,s_2,h,f)$ is a local coordinate system on $U$.
Assume moreover that the symplectic Poisson--Nijenhuis recursion operator of \Cref{prop:Kow-main} is defined on $U$
\[
N:TU\to TU
\]
whose dual map $N^*:T^*U\to T^*U$ leaves invariant the two rank--$2$ subbundles
\[
N^*\big(\mathrm{span}\{ds_1,ds_2\}\big)\subset \mathrm{span}\{ds_1,ds_2\},
\qquad
N^*\big(\mathrm{span}\{dh,df\}\big)\subset \mathrm{span}\{dh,df\}.
\]
In addition, assume that the restriction of $N^*$ to $\mathrm{span}\{ds_1,ds_2\}$
is represented in the coframe $(ds_1,ds_2)$ by the transpose companion matrix acting on row components, i.e.
\begin{equation}\label{eq:Nstar-companion}
N^*ds_1=s_1\,ds_1+ds_2,\qquad N^*ds_2=s_2\,ds_1,
\end{equation}
and that the restriction of $N^*$ to $\mathrm{span}\{dh,df\}$ is given by
\begin{equation}\label{eq:Nstar-hf}
N^*dh=m_1\,dh+m_2\,df,\qquad N^*df=m_3\,dh+m_4\,df,
\end{equation}
for some smooth functions $m_1,m_2,m_3,m_4\in C^\infty(U)$ (e.g.\ the entries of the
control matrix in the frame $(X_h,X_f)$ on $\mathcal D$).
\noindent
Equivalently, in the coframe \((ds_1,ds_2,dh,df)\), the action of \(N^*\) on row components is given by right multiplication by
\[
[N^*]_{(ds_1,ds_2,dh,df)}=
\begin{pmatrix}
s_1 & 1   & 0 & 0\\
s_2 & 0   & 0 & 0\\
0   & 0   & m_1 & m_2\\
0   & 0   & m_3 & m_4
\end{pmatrix}
\]
so that the separation sector \(\mathrm{span}\{ds_1,ds_2\}\) and the integrals
sector \(\mathrm{span}\{dh,df\}\) are \(N^*\)-invariant and decouple.

%%%%%%%%%%%%%%%%%%%%%%%%%%%%%%%%%%%%%%%%%%%%%%%%%%%%
%%
%%
%%
%%
%%
%%
%%%%%%%%%%%%%%%%%%%%%%%%%%%%%%%%%%%%%%%%%%%%%%%%%%%%

\subsection{Separation coordinates and conjugate momenta}
Assume moreover that
\[
du^1\wedge du^2\wedge dH_0\wedge dH_1\neq 0
\]
on a possibly smaller open set, so that \((u^1,u^2,H_0,H_1)\) are local coordinates. Differentiating
\[
s_1=u^1+u^2,
\qquad
s_2=-u^1u^2
\]
gives
\[
ds_1=du^1+du^2,
\qquad
ds_2=-(u^2\,du^1+u^1\,du^2),
\]
hence
\begin{equation}\label{eq:du-in-ds}
du^1=\dfrac {u^1\,ds_1+ds_2}{u^1-u^2},
\qquad
du^2=\dfrac {u^2\,ds_1+ds_2}{u^2-u^1}.
\end{equation}

Let \(Y_i\in\mathfrak X(U)\), \(i=1,2\), be the leaf-tangent vector fields defined by
\[
du^j(Y_i)=\delta_i^{\,j},
\qquad
dH_a(Y_i)=0,
\qquad
j=1,2,\ a=0,1.
\]
On each regular common level \(F_c\), define the coefficients of the restricted Liouville form by
\begin{equation}\label{eq:pi-def}
\pi_i:=\theta(Y_i),
\qquad i=1,2,
\end{equation}
where
\[
\theta=p_x\,dx+p_y\,dy
\]
is the Liouville \(1\)-form. Equivalently,
\begin{equation}\label{eq:pi-explicit}
\pi_i
=
p_x\,Y_i(x)+p_y\,Y_i(y)
=
p_x\,\dfrac {\partial x}{\partial u^i}\Big|_{H}
+
p_y\,\dfrac {\partial y}{\partial u^i}\Big|_{H},
\qquad i=1,2.
\end{equation}

Writing \(\iota_c:F_c\hookrightarrow M\), these formulas give
\[
\iota_c^*\theta=\pi_1\,du^1+\pi_2\,du^2,\qquad
d(\iota_c^*\theta)=\iota_c^*\omega_0=0.
\]
Thus \(\iota_c^*\theta\) is locally the differential of a generating function on the Lagrangian leaf. This does not make \((u^1,u^2,\pi_1,\pi_2)\) an ambient canonical chart: the coefficients \(\pi_i=\theta(Y_i)\), even when extended to \(U\), need not satisfy \(\{u^i,\pi_j\}_0=\delta_{ij}\).

Ambient Darboux--Nijenhuis momenta \(v_i\) instead exist under \Cref{prop:Kow-main}. In each explicit construction below they are verified through
\[
\{u^i,u^j\}_0=\{v_i,v_j\}_0=0,\qquad
\{u^i,v_j\}_0=\delta_{ij},\qquad
N^*dv_i=u^i\,dv_i.
\]
A shift \(v_i\mapsto v_i+\partial_{u^i}G\) preserves the canonical brackets. It preserves the Darboux--Nijenhuis normalization for distinct eigenvalues only if \(\partial_{u^1}\partial_{u^2}G=0\), locally, so the remaining shifts depend on the individual eigenvalue alone.

%%%%%%%%%%%%%%%%%%%%%%%%%%%%%%%%%%%%%%%%%%%%%%%%%%%%
%%%%%
%%%%%
%%%%%
%%%%%
%%%%%
%%%%%
%%%%%%%%%%%%%%%%%%%%%%%%%%%%%%%%%%%%%%%%%%%%%%%%%%%%

\section{Integrable cubic H\'enon--Heiles systems}

In this section we examine the three integrable generalized cubic cases of the H\'enon--Heiles Hamiltonian from the viewpoint of the bi-Hamiltonian approach. The explicit constructions include arbitrary \(\mu\) for KK and KdV$_5$, and \(\mu=0\) for SK. All spectral and canonical charts are local, with the restrictions specified below. Although the overall strategy is essentially the same for all three systems, we present each of them in full detail, since they exhibit significant geometric and structural differences.

%%%%%%%%%%%%%%%%%%%%%%%%%%%%%%%%%%%%%%%%%%%%%%%%%%%%
%%%%%
%%%%%
%%%%%
%%%%%
%%%%%
%%%%%
%%%%%%%%%%%%%%%%%%%%%%%%%%%%%%%%%%%%%%%%%%%%%%%%%%%%
\subsection{Case (iii): the Kaup--Kupershmidt (KK) system}\label{subsec:kk}

We now turn to the Kaup--Kupershmidt case. In this subsection we carry out the bi-Hamiltonian construction for this system, including the explicit commuting integrals, the deformation vector field, the deformed Poisson tensor, the recursion operator, the control matrix, the separating coordinates, and the corresponding conjugate momenta.

We work on \(y\ne0\) with the canonical symplectic structure, in coordinates \((p_x,p_y,x,y)\). The Hamiltonian is
\begin{equation}\label{eq:kk-H}
H_{\mathrm{KK}}
=
\dfrac12\bigl(p_x^2+p_y^2\bigr)
+8\omega x^2
+\dfrac{\omega}{2}y^2
+a\,x y^2
+\dfrac{16a}{3}x^3
+\dfrac{\mu}{2y^2}.
\end{equation}
This is precisely the Kaup--Kupershmidt member of the integrable cubic H\'enon--Heiles family discussed above. The corresponding second independent first integral is
\begin{equation}\label{eq:K-KK}
K_{\mathrm{KK}}
=
9\left(\omega y^2+p_y^2+\dfrac{\mu}{y^2}\right)^2
+12a\,p_y\,y^2\,(3x p_y-y p_x)
-2a^2y^4(6x^2+y^2)
+12ax(\mu-\omega y^4)
-18\omega \mu .
\end{equation}
A direct, although somewhat lengthy, computation with the canonical Poisson bracket shows that
\begin{equation}\label{eq:kk-involution}
\{H_{\mathrm{KK}},K_{\mathrm{KK}}\}_0=0.
\end{equation}
This verification can also be carried out efficiently with a computer algebra system such as Maple. Hence the pair \(\bigl(H_{\mathrm{KK}},K_{\mathrm{KK}}\bigr)\) defines a Liouville integrable Hamiltonian system.

The two integrals are invariant under the symplectic involution
\begin{equation}\label{eq:kk-sigma}
\sigma:(p_x,p_y,x,y)\longmapsto (p_x,-p_y,x,-y).
\end{equation}
Taking into account the \(\sigma\)-invariance of the pair \((H_{\mathrm{KK}},K_{\mathrm{KK}})\), together with the additional symmetries guiding our ansatz, and with the aid of Maple, we are led to the deformation vector field
\begin{equation}\label{eq:kk-X}
X=
X^{p_x}\dfrac{\partial}{\partial p_x}
+
X^{p_y}\dfrac{\partial}{\partial p_y}
+
X^{x}\dfrac{\partial}{\partial x}
+
X^{y}\dfrac{\partial}{\partial y},
\end{equation}
where, in the canonical coordinates \((p_x,p_y,x,y)\),
\begin{equation}\label{eq:kk-X-components}
X=
\begin{pmatrix}
-4ax\,p_x-6ay\,p_y+\dfrac{48x(ax+\omega)}{y}\,p_y\\[0.8em]
0\\[0.8em]
-\dfrac12 ay^2-3\,\dfrac{p_xp_y}{y}-3\omega x\\[0.8em]
-6\,\dfrac{p_y^2}{y}+(8ax+3\omega)y
\end{pmatrix}.
\end{equation}

Accordingly, the deformed Poisson tensor is defined by
\begin{equation}\label{eq:kk-P1-def}
P_1=\mathcal L_XP_0,
\end{equation}
where, with respect to the ordering \((p_x,p_y,x,y)\), the canonical Poisson tensor is
\begin{equation}\label{eq:kk-P0}
P_0=
\begin{pmatrix}
0&0&-1&0\\
0&0&0&-1\\
1&0&0&0\\
0&1&0&0
\end{pmatrix}.
\end{equation}
A direct computation yields
\begin{equation}\label{eq:kk-P1}
P_1=
\begin{pmatrix}
0 &
\dfrac{6p_y\bigl(ay^2+8x(ax+\omega)\bigr)}{y^2} &
-(4ax+3\omega) &
\dfrac{2\bigl(ay^2+24x(ax+\omega)\bigr)}{y}
\\[1em]
-\dfrac{6p_y\bigl(ay^2+8x(ax+\omega)\bigr)}{y^2} &
0 &
-ay+\dfrac{3p_xp_y}{y^2} &
8ax+3\omega+\dfrac{6p_y^2}{y^2}
\\[1em]
4ax+3\omega &
ay-\dfrac{3p_xp_y}{y^2} &
0 &
-\dfrac{3p_x}{y}
\\[1em]
-\dfrac{2\bigl(ay^2+24x(ax+\omega)\bigr)}{y} &
-\left(8ax+3\omega+\dfrac{6p_y^2}{y^2}\right) &
\dfrac{3p_x}{y} &
0
\end{pmatrix}.
\end{equation}

The recursion operator associated with the deformation tensor \(P_1\) is
\begin{equation}\label{eq:kk-N}
N=P_1P_0^{-1}=-P_1P_0.
\end{equation}
A direct computation gives
\begin{equation}\label{eq:kk-N-matrix}
N=
\begin{pmatrix}
4ax+3\omega &
-\dfrac{2\bigl(ay^2+24x(ax+\omega)\bigr)}{y} &
0 &
\dfrac{6p_y\bigl(ay^2+8x(ax+\omega)\bigr)}{y^2}
\\[1em]
ay-\dfrac{3p_xp_y}{y^2} &
-\left(8ax+3\omega+\dfrac{6p_y^2}{y^2}\right) &
-\dfrac{6p_y\bigl(ay^2+8x(ax+\omega)\bigr)}{y^2} &
0
\\[1em]
0 &
\dfrac{3p_x}{y} &
4ax+3\omega &
ay-\dfrac{3p_xp_y}{y^2}
\\[1em]
-\dfrac{3p_x}{y} &
0 &
-\dfrac{2\bigl(ay^2+24x(ax+\omega)\bigr)}{y} &
-\left(8ax+3\omega+\dfrac{6p_y^2}{y^2}\right)
\end{pmatrix}.
\end{equation}
The entries of \(N\) are sufficiently involved that checking the vanishing of the Nijenhuis torsion by hand is not illuminating. Using Maple, one verifies directly from the coordinate formula for \(T_N\) that all its components vanish identically. Therefore \(N\) is a torsionless recursion operator.

It is worth observing that the deformation vector field \(X\), and hence the tensors \(P_1=\mathcal L_XP_0\) and \(N=P_1P_0^{-1}\), do not depend explicitly on the parameter \(\mu\). This is not contradictory: \(P_1\) and \(N\) are ambient geometric objects determined solely by the deformation of the canonical Poisson tensor. By contrast, the control matrix introduced below describes the action of \(N^*\) on the \(\mu\)-dependent codistribution \(\operatorname{span}\{dH_{\mathrm{KK}},dK_{\mathrm{KK}}\}\), and therefore it does depend on \(\mu\).

We now compute the control matrix \(M\) associated with the recursion operator \(N\). By definition, \(M\) is determined by the relations
\begin{equation}\label{eq:kk-control-def}
N^* dH_{\mathrm{KK}}=M_{11}\,dH_{\mathrm{KK}}+M_{12}\,dK_{\mathrm{KK}},
\qquad
N^* dK_{\mathrm{KK}}=M_{21}\,dH_{\mathrm{KK}}+M_{22}\,dK_{\mathrm{KK}}.
\end{equation}
A direct symbolic computation shows that the span of \(\{dH_{\mathrm{KK}},dK_{\mathrm{KK}}\}\) is invariant under \(N^*\) and yields
\begin{equation}\label{eq:kk-M-general}
M=
\begin{pmatrix}
-\;2ax-\dfrac{3p_y^2}{y^2}+\dfrac{3\mu}{y^4}
&
-\dfrac{1}{12y^2}
\\[1.1em]
-\dfrac{12}{y^2}\bigl(K_{\mathrm{KK}}+6\mu\,M_{22}\bigr)
&
-\;2ax-\dfrac{3p_y^2}{y^2}-\dfrac{3\mu}{y^4}
\end{pmatrix},
\end{equation}
where
\begin{equation}\label{eq:kk-M22-general}
M_{22}=-2ax-\dfrac{3p_y^2}{y^2}-\dfrac{3\mu}{y^4}.
\end{equation}

The separating coordinates are obtained as the eigenvalues of the control matrix \(M\). To compute them, it is convenient to introduce
\begin{equation}\label{eq:kk-alpha-beta}
\alpha:=-2ax-\dfrac{3p_y^2}{y^2},
\qquad
\beta:=\dfrac{3\mu}{y^4},
\end{equation}
so that
\begin{equation}\label{eq:kk-M-compact}
M=
\begin{pmatrix}
\alpha+\beta & -\dfrac{1}{12y^2}\\[1em]
-\dfrac{12}{y^2}\bigl(K_{\mathrm{KK}}+6\mu(\alpha-\beta)\bigr) & \alpha-\beta
\end{pmatrix}.
\end{equation}
Hence
\begin{equation}\label{eq:kk-trace-det}
\operatorname{tr}M=2\alpha=-4ax-\dfrac{6p_y^2}{y^2},
\end{equation}
and
\begin{equation}\label{eq:kk-det}
\det M
=
(\alpha+\beta)(\alpha-\beta)
-\dfrac{1}{y^4}\bigl(K_{\mathrm{KK}}+6\mu(\alpha-\beta)\bigr)
=
(\alpha-\beta)^2-\dfrac{K_{\mathrm{KK}}}{y^4}.
\end{equation}
Therefore the characteristic polynomial is
\begin{equation}\label{eq:kk-charpoly-M}
\chi_M(\lambda)=\det(\lambda I-M)
=
\lambda^2
+\left(4ax+\dfrac{6p_y^2}{y^2}\right)\lambda
+
\left(2ax+\dfrac{3p_y^2}{y^2}+\dfrac{3\mu}{y^4}\right)^2
-\dfrac{K_{\mathrm{KK}}}{y^4}.
\end{equation}
Its roots are
\begin{equation}\label{eq:kk-u-general}
u_{1,2}
=
-2ax-\dfrac{3p_y^2}{y^2}
\pm
\dfrac{1}{y^2}
\sqrt{
K_{\mathrm{KK}}
-12a\mu x
-\dfrac{18\mu p_y^2}{y^2}
-\dfrac{9\mu^2}{y^4}
}.
\end{equation}
We use these eigenvalues on a local chart with \(a y\Delta\ne0\). For real coordinates we require \(\Delta^2>0\) and choose a smooth square-root branch; otherwise the construction is locally complex. Such a real spectral chart is not asserted to be dense in the entire real phase domain.

\subsubsection*{Separating coordinates and conjugate momenta in the general case}

For later convenience, let us introduce the shorthand
\begin{equation}\label{eq:kk-Delta}
\Delta:=
\sqrt{
K_{\mathrm{KK}}
-12a\mu x
-\dfrac{18\mu p_y^2}{y^2}
-\dfrac{9\mu^2}{y^4}
},
\end{equation}
A direct expansion of \eqref{eq:K-KK} shows that
\[
K_{\mathrm{KK}}=K_0+12a\mu x+\frac{18\mu p_y^2}{y^2}+\frac{9\mu^2}{y^4},
\qquad K_0:=K_{\mathrm{KK}}|_{\mu=0},\qquad \Delta^2=K_0.
\]
In particular, the eigenvalues and momenta below are independent of \(\mu\). We have
\begin{equation}\label{eq:kk-u-Delta}
u_{1,2}=-2ax-\dfrac{3p_y^2}{y^2}\pm\dfrac{\Delta}{y^2}.
\end{equation}

We now construct conjugate momenta \(v_1,v_2\), following our convention of writing canonical pairs with the momenta first, namely
\begin{equation}\label{eq:kk-pairs-general}
(v_1,u_1,v_2,u_2).
\end{equation}
A convenient choice is
\begin{equation}\label{eq:kk-v-general-ansatz}
v_i=-\dfrac{p_x}{4a}-\dfrac{p_y}{4a^2y}\bigl(u_i-4ax-3\omega\bigr),
\qquad i=1,2.
\end{equation}
Equivalently,
\begin{equation}\label{eq:kk-v-general}
v_{1,2}
=
-\dfrac{p_x}{4a}
-
\dfrac{p_y}{4a^2y}
\left(
-6ax-3\omega-\dfrac{3p_y^2}{y^2}
\pm\dfrac{\Delta}{y^2}
\right).
\end{equation}

The dependence on \(p_x\) must be retained: differentiating \(\Delta^2=K_0\) gives
\begin{equation}\label{eq:kk-u-px}
\frac{\partial u_1}{\partial p_x}=-\frac{6ayp_y}{\Delta},\qquad
\frac{\partial u_2}{\partial p_x}=\frac{6ayp_y}{\Delta}.
\end{equation}
We write the candidate momenta as
\begin{equation}\label{eq:kk-v-affine}
v_i=-\frac{p_x}{4a}+\phi_i(p_x,p_y,x,y),
\end{equation}
and verify the canonical conditions
\begin{equation}\label{eq:kk-canonical-target}
\{u_i,v_j\}_0=\delta_{ij}
\end{equation}
using
\begin{equation}\label{eq:kk-poisson-bracket}
\{f,g\}_0=f_xg_{p_x}-f_{p_x}g_x+f_yg_{p_y}-f_{p_y}g_y.
\end{equation}
For the choice
\begin{equation}\label{eq:kk-phi-general}
\phi_i=-\frac{p_y}{4a^2y}(u_i-4ax-3\omega),
\end{equation}
the derivative of \(u_i\) in \(\phi_i\) is included. Substitution of \(d\Delta=dK_0/(2\Delta)\), followed by reduction using \(\Delta^2=K_0\), gives all the identities below; this supplies a direct verification of the candidate rather than an affine-in-\(p_x\) derivation.

A direct symbolic verification then shows that the full set of canonical relations is satisfied:
\begin{equation}\label{eq:kk-canonical-general}
\{u_i,u_j\}_0=0,
\qquad
\{u_i,v_j\}_0=\delta_{ij},
\qquad
\{v_i,v_j\}_0=0.
\end{equation}
The same differentiation verifies
\[
N^*du_i=u_i\,du_i,\qquad N^*dv_i=u_i\,dv_i,\qquad i=1,2.
\]
Consequently these functions define a local Darboux--Nijenhuis chart on the stated domain. Exact symbolic checks use only rational operations and the relation \(\Delta^2=K_0\).

Reversing this change of coordinates we can express the Hamiltonian functions in the new coordinates:
\begin{equation}\label{h1h2}
\begin{aligned}
h_1&=\frac{F_1+F_2}{2}-\frac{\mu(u_1-u_2)}{12d},\\
h_2&=\frac{d^2}{4}-\frac{\mu(u_1+u_2)}{12}
       +\frac{\mu^2(u_1-u_2)^2}{144d^2},
\qquad F_i=f(v_i,u_i),\quad d=F_1-F_2,
\end{aligned}
\end{equation}
where
\[
f(v,u) = \dfrac{96 a^{4} v^{2} +27 \omega^{3}-9 \omega^{2} u -3 \omega  u^{2}+u^{3}}{12 a^{2}}.
\]
The normalization is precisely \(h_1=H_{\mathrm{KK}}\), \(h_2=K_{\mathrm{KK}}/36\). Useful exact identities are
\[
F_1-F_2=-\Delta/3,\qquad (F_1+F_2)/2=H_{\mathrm{KK}}|_{\mu=0}.
\]
Thus \(d\ne0\) on our chart, and \(K\) below denotes the value of the normalized integral \(h_2\). Fixing the values of the Hamiltonian functions
\[
h_1 = E,
\qquad
h_2 = K,
\]
the expression for \(h_1\) can easily be separated in the form
\[
(f(v_1,u_1) - E)^2 - \dfrac{\mu u_1}{6}
=
(f(v_2,u_2) - E)^2 - \dfrac{\mu u_2}{6}
=
C.
\]
Using \eqref{h1h2}, one realizes that this constant is \(C=K\), so that the separated equations can be written
\begin{equation}\label{sep_kk}
(f(v_1, u_1) - E)^2 - \dfrac{\mu u_1}{6} = K = (f(v_2, u_2) - E)^2 - \dfrac{\mu u_2}{6}\,.
\end{equation}
Solving \eqref{sep_kk} involves independent local choices
\[
F_i-E=\varepsilon_i\sqrt{K+\mu u_i/6},\qquad \varepsilon_i\in\{1,-1\},
\]
and a further sign for each momentum:
\[
v_i^2=\frac{12a^2E-27\omega^3+9\omega^2u_i+3\omega u_i^2-u_i^3
 +12a^2\varepsilon_i\sqrt{K+\mu u_i/6}}{96a^4}.
\]
The signs are continued along the chosen branches of the separated curve. Differentiating \eqref{h1h2} gives the opposite coupling signs
\begin{equation}\label{eq:kk-u-system}
\dot u_1=8a^2v_1\left(1+\frac{\mu(u_1-u_2)}{6d^2}\right),\qquad
\dot u_2=8a^2v_2\left(1-\frac{\mu(u_1-u_2)}{6d^2}\right).
\end{equation}
In general these equations do not have a common prefactor multiplying two one-variable radicals.

Since our convention is to write canonical pairs with the momenta first, namely
\[
(v_1,u_1,v_2,u_2),
\]
the Hamilton equations for a Hamiltonian \(\mathcal H(v_1,u_1,v_2,u_2)\) take the form
\begin{equation}\label{eq:kk-hamilton-system}
\left\{
\begin{aligned}
\dot v_1 &= -\dfrac{\partial \mathcal H}{\partial u_1}\\[4pt]
\dot u_1 &= \phantom{-}\dfrac{\partial \mathcal H}{\partial v_1}
\end{aligned}
\right.
\hskip2truecm
\left\{
\begin{aligned}
\dot v_2 &= -\dfrac{\partial \mathcal H}{\partial u_2}\\[4pt]
\dot u_2 &= \phantom{-}\dfrac{\partial \mathcal H}{\partial v_2}.
\end{aligned}
\right.
\end{equation}
For \(\mu\ne0\), fix a regular level \(h_1=E\), \(h_2=K\), and introduce signed variables
\[
z_i=\sqrt6(F_i-E),\qquad z_i^2=\mu u_i+6K,\qquad
u_i=\frac{z_i^2-6K}{\mu},\qquad d=\frac{z_1-z_2}{\sqrt6}.
\]
Then \eqref{eq:kk-u-system} gives
\[
\dot u_i=\frac{16a^2v_i z_i}{z_i-z_j},\qquad
\dot z_i=\frac{8\mu a^2v_i}{z_i-z_j},\qquad j\ne i.
\]
The second identity also follows directly from \(\{\sqrt6(F_i-h_1),h_1\}_0\), so it remains valid at \(z_i=0\) without dividing by \(z_i\). Define
\begin{align*}
B_6(z)={}&12a^2E-27\omega^3+\frac{9\omega^2}{\mu}(z^2-6K)
 +\frac{3\omega}{\mu^2}(z^2-6K)^2\\
&-\frac{1}{\mu^3}(z^2-6K)^3+2\sqrt6\,a^2z,\qquad
R_6(z)=\frac{2\mu^2}{3}B_6(z).
\end{align*}
With \(W_i=8\mu a^2v_i\) one has
\[
W_i^2=R_6(z_i),\qquad
\dot z_1=\frac{W_1}{z_1-z_2},\qquad
\dot z_2=\frac{W_2}{z_2-z_1}.
\]
Consequently the orbit and its original time parametrization are determined by
\begin{equation}\label{eq:kk-quadrature-relation}
\frac{dz_1}{W_1}+\frac{dz_2}{W_2}=0,\qquad
\frac{z_1\,dz_1}{W_1}+\frac{z_2\,dz_2}{W_2}=dt.
\end{equation}
Locally, the sum of the first two Abelian integrals is constant, and the weighted sum equals \(t-t_0\), with lower limits at the initial points. The signs of \(W_i\) are fixed by the actual momenta; at a zero of \(W_i\) the differentials are interpreted on the curve. A square-free sextic gives a nonsingular curve of genus two; special levels may degenerate. This is consistent with the hyperelliptic integration in \cite{VMC02} and the spectral interpretation in \cite{EFH15}.

\begin{remark}
For \(\mu=0\), start directly from \eqref{h1h2}: \(h_1=(F_1+F_2)/2\), \(h_2=d^2/4\), and \(\dot u_i=8a^2v_i\). Thus the polynomial KK Hamiltonian is additively separated. On a regular level \(d\ne0\), the signs of \(F_1-E\) and \(F_2-E\) are opposite when expressed through a common square root of \(K\). No common-factor formula with a vanishing radical denominator is needed, and the \(\mu\ne0\) covering is not used in this limit.
\end{remark}
\begin{remark}
The nested radicals obtained by solving for \(v\) can be described by the same algebraic change of variables. The original implicit separated relation is already algebraic.

Indeed, starting from the separated relation
\begin{equation}\label{eq:kk-separated-covering-start}
(f(v,u)-E)^2-\frac{\mu u}{6}=K,
\end{equation}
we introduce a new variable \(z\) by
\begin{equation}\label{eq:kk-covering}
z^2=\mu u+6K.
\end{equation}
Assuming \(\mu\neq 0\), this is equivalent to
\begin{equation}\label{eq:kk-u-covering}
u=\frac{z^2-6K}{\mu}.
\end{equation}
Substituting \eqref{eq:kk-covering} into \eqref{eq:kk-separated-covering-start}, we obtain
\begin{equation}\label{eq:kk-separated-z}
(f(v,u)-E)^2=\frac{z^2}{6}.
\end{equation}
Hence, on the covering, the separated relation can be written as
\begin{equation}\label{eq:kk-f-minus-E}
f(v,u)-E=\frac{z}{\sqrt6},\qquad z:=\sqrt6(f(v,u)-E).
\end{equation}

Recalling that
\begin{equation}\label{eq:kk-fvu-covering}
f(v,u)=\frac{96a^4v^2+27\omega^3-9\omega^2u-3\omega u^2+u^3}{12a^2},
\end{equation}
equation \eqref{eq:kk-f-minus-E} yields
\begin{equation}\label{eq:kk-v2-before-subst}
96a^4v^2
=
12a^2E-27\omega^3+9\omega^2u+3\omega u^2-u^3
+2\sqrt6\,a^2 z.
\end{equation}
Using \eqref{eq:kk-u-covering}, we then obtain
\begin{equation}\label{eq:kk-v2-after-subst}
96a^4v^2
=
12a^2E-27\omega^3
+\frac{9\omega^2}{\mu}(z^2-6K)
+\frac{3\omega}{\mu^2}(z^2-6K)^2
-\frac{1}{\mu^3}(z^2-6K)^3
+2\sqrt6\,a^2 z.
\end{equation}
Equivalently, after multiplying by \(\mu^3\), one gets
\begin{equation}\label{eq:kk-hyperelliptic-raw}
96a^4\mu^3 v^2
=
12a^2E\,\mu^3-27\omega^3\mu^3
+9\omega^2\mu^2(z^2-6K)
+3\omega\mu (z^2-6K)^2
-(z^2-6K)^3
+2\sqrt6\,a^2\mu^3 z.
\end{equation}

Therefore the separated curve has the polynomial model
\begin{equation}\label{eq:kk-hyperelliptic-form}
v^2=\Pi_6(z),\qquad \Pi_6(z)=\frac{B_6(z)}{96a^4}.
\end{equation}
For \(\mu\ne0\), the formulas \(z=\sqrt6(f(v,u)-E)\) and \(u=(z^2-6K)/\mu\) give inverse algebraic changes on the separated curve. The square-root projection \(z^2=\mu u+6K\) is two-sheeted over the \(u\)-coordinate; the signed change is not a two-to-one map onto the full implicit separated curve. Replacing \(z\) by \(-z\) gives the equivalent convention with a minus sign in the odd term.

If \(K\) is held as an external parameter on an auxiliary two-dimensional phase plane, the local point transformation \(u=(z^2-6K)/\mu\), \(z\ne0\), is canonical with
\[
w=\frac{2z}{\mu}v,\qquad v\,du=w\,dz.
\]
This is not an ambient four-dimensional canonical transformation for the actual system, where \(K=h_2\) varies. Indeed, extending \(z_i=\sqrt6(F_i-h_1)\) to phase space gives
\[
\{z_1,z_2\}_0=-\frac{8\mu a^2(v_1+v_2)}{d},\qquad
\sum_{i=1}^2\frac{2z_iv_i}{\mu}\,dz_i
=\sum_{i=1}^2v_i\,du_i+\frac{6}{\mu}(v_1+v_2)\,dh_2.
\]
The variables \((z_i,W_i)\) above are used for the level-curve quadratures, not as ambient canonical pairs.
\end{remark}

\begin{example}
A particularly simple but still nontrivial example is obtained by choosing
\[
a=1,\qquad \mu=1,\qquad \omega=0,\qquad K=0.
\]
Then the separated relation
\[
(f(v,u)-E)^2-\frac{\mu u}{6}=K
\]
reduces to
\[
(f(v,u)-E)^2=\frac{u}{6}.
\]
Since in this case
\[
f(v,u)=\frac{96v^2+u^3}{12},
\]
we obtain
\[
\left(\frac{96v^2+u^3}{12}-E\right)^2=\frac{u}{6}.
\]

Introducing the covering variable (the two signs below are exchanged by \(z\mapsto-z\))
\[
z^2=u,
\]
the separated relation becomes
\[
\frac{96v^2+z^6}{12}-E=\pm \frac{z}{\sqrt6},
\]
that is,
\[
96v^2=12E-z^6\pm 2\sqrt6\,z.
\]
Therefore, on the covering, the separated curve takes the polynomial hyperelliptic form
\[
v^2=\frac{1}{96}\bigl(12E-z^6\pm 2\sqrt6\,z\bigr).
\]

This example shows explicitly that the irrationality appearing in the original separated variable \(u\) is removed by passing to the natural covering \(z^2=u\), and that the already algebraic separated curve admits a polynomial hyperelliptic model in \(z\).
\end{example}

% \begin{example}\label{ex:kk-omega-zero-classical}
% As a particularly simple nontrivial example, let
% \[
% \omega=0.
% \]
% Then \eqref{eq:kk-H} and \eqref{eq:K-KK}  become
% \[
% H_{\mathrm{KK}}
% =
% \dfrac12(p_x^2+p_y^2)+a\,xy^2+\dfrac{16a}{3}x^3,
% \]
% \[
% K_{\mathrm{KK}}
% =
% 9p_y^4+12a\,y^2p_y(3x\,p_y-y\,p_x)-2a^2y^4(6x^2+y^2).
% \]
% The separating coordinates reduce to
% \[
% u_1=\dfrac{\sqrt{K_{\mathrm{KK}}}}{y^2}-\left(3\dfrac{p_y^2}{y^2}+2ax\right),
% \qquad
% u_2=-\dfrac{\sqrt{K_{\mathrm{KK}}}}{y^2}-\left(3\dfrac{p_y^2}{y^2}+2ax\right),
% \]
% and the momenta \eqref{eq:kk-v-classical-explicit} become
% \[
% v_1=\dfrac{1}{16a^2}\left(-4ap_x-2\dfrac{p_y}{y}u_2\right),
% \qquad
% v_2=\dfrac{1}{16a^2}\left(-4ap_x-2\dfrac{p_y}{y}u_1\right).
% \]
% In these variables the Hamiltonian separates as
% \[
% H_{\mathrm{KK}}=H_1(v_1,u_1)+H_2(v_2,u_2),
% \qquad
% H_i(v_i,u_i)=4a^2v_i^2+\dfrac{u_i^3}{24a^2}.
% \]
% The Hamilton equations reduce to
% \[
% \dot v_i=-\dfrac{u_i^2}{8a^2},
% \qquad
% \dot u_i=8a^2v_i,
% \qquad i=1,2,
% \]
% or equivalently,
% \[
% \ddot u_i+u_i^2=0,
% \qquad i=1,2.
% \]
% Thus the dynamics splits into two independent one-degree-of-freedom systems.
% \end{example}

%%%%%%%%%%%%%%%%%%%%%%%%%%%%%%%%%%%%%%%%%%%%%%%%%%%%
%%%%%
%%%%%
%%%%%
%%%%%
%%%%%
%%%%%
%%%%%%%%%%%%%%%%%%%%%%%%%%%%%%%%%%%%%%%%%%%%%%%%%%%%

\subsection{Case (ii): the KdV$_5$ system}

This subsection is devoted to the KdV$_5$ case. Our aim is to reproduce, as far as possible, the same bi-Hamiltonian construction developed in detail for the Kaup--Kupershmidt system, namely the determination of the relevant integrals, the deformed Poisson tensor, the recursion operator, the control matrix, the separating coordinates, and the corresponding conjugate momenta. Since the general strategy is by now standard, we emphasize the points specific to the present case and refer to the previous subsection for those computations that follow the same pattern verbatim.

We work on the natural phase domain with canonical coordinates \((p_x,p_y,x,y)\), and we consider the generalized cubic H\'enon--Heiles Hamiltonian
\[
H_{\mathrm{KdV}_5}
=
\dfrac12\bigl(p_x^2+p_y^2\bigr)
+\dfrac12\bigl(\omega_1x^2+\omega_2y^2\bigr)
+a\,x y^2+2a\,x^3+\dfrac{\mu}{2y^2}.
\]
This is precisely the KdV$_5$ integrable case in our notation. A convenient representative of the second independent first integral is
\[
K_{\mathrm{KdV}_5}
=
4ay\,p_xp_y
+\bigl(4\omega_2-\omega_1-4ax\bigr)
\left(p_y^2+\dfrac{\mu}{y^2}\right)
+a^2y^4
+4a^2x^2y^2
+4a\omega_2xy^2
+\omega_2(4\omega_2-\omega_1)y^2.
\]
A direct computation with the canonical Poisson bracket shows that
\[
\{H_{\mathrm{KdV}_5},K_{\mathrm{KdV}_5}\}_0=0.
\]
Hence the pair \(\bigl(H_{\mathrm{KdV}_5},K_{\mathrm{KdV}_5}\bigr)\) defines a completely integrable Hamiltonian system in the Liouville sense.

As in the Kaup--Kupershmidt case, the starting point of the bi-Hamiltonian construction is the search for a vector field \(X\) such that
\[
P_1=\mathcal L_XP_0
\]
defines a second Poisson tensor compatible with the canonical Poisson tensor \(P_0\), and such that the associated recursion operator
\[
N=P_1P_0^{-1}
\]
is torsionless. Once such a deformation has been found, one may extract the separating coordinates from the two distinct eigenvalues of \(N\), and then recover the conjugate separated momenta by imposing the Darboux--Nijenhuis conditions exactly as in the previous subsection.

In the present case we follow the same sequence of steps:
\[
\bigl(H_{\mathrm{KdV}_5},K_{\mathrm{KdV}_5}\bigr)
\quad\Longrightarrow\quad
X
\quad\Longrightarrow\quad
P_1=\mathcal L_XP_0
\quad\Longrightarrow\quad
N=P_1P_0^{-1}
\quad\Longrightarrow\quad
(u_1,u_2)
\quad\Longrightarrow\quad
(v_1,v_2).
\]
We now describe the main ingredients of this construction.

\medskip
\noindent
\textbf{Control matrix and the explicit separated variables.}
As in the Kaup--Kupershmidt case, we look for a deformation of the canonical Poisson tensor of the form
\[
P_1=\mathcal L_XP_0
\]
such that the associated recursion operator
\[
N=P_1P_0^{-1}
\]
is torsionless and its two distinct eigenvalues provide the separating coordinates. In the present case it is convenient to set
\[
\delta=\dfrac{4\omega_2-\omega_1}{2a},
\]
and to consider the vector field
\[
X=
\begin{pmatrix}
-(2x-\delta)p_x-y\,p_y\\[6pt]
-y\,p_x\\[6pt]
0\\[6pt]
0
\end{pmatrix}.
\]
Since \(P_0\) is constant, one has
\[
P_1=\mathcal L_XP_0=-(J_XP_0+P_0J_X^\top),
\]
where \(J_X\) denotes the Jacobian matrix of \(X\). A direct computation yields
\[
P_1=
\begin{pmatrix}
0 & p_y & \delta-2x & -y\\[4pt]
-p_y & 0 & -y & 0\\[4pt]
2x-\delta & y & 0 & 0\\[4pt]
y & 0 & 0 & 0
\end{pmatrix},
\]
and therefore
\[
N=P_1P_0^{-1}=
\begin{pmatrix}
2x-\delta & y & 0 & p_y\\[4pt]
y & 0 & -p_y & 0\\[4pt]
0 & 0 & 2x-\delta & y\\[4pt]
0 & 0 & y & 0
\end{pmatrix}.
\]
Its characteristic polynomial is
\[
\chi_N(\lambda)=\det(\lambda I-N)
=
\bigl(\lambda^2-(2x-\delta)\lambda-y^2\bigr)^2,
\]
so that the Darboux--Nijenhuis coordinates are
\[
u_{1,2}
=
\dfrac{2x-\delta\pm\sqrt{(2x-\delta)^2+4y^2}}{2},
\qquad
u_1+u_2=2x-\delta,
\qquad
u_1u_2=-y^2.
\]
We now consider the codistribution
\[
\operatorname{span}\{dH_{\mathrm{KdV}_5},\,dK_{\mathrm{KdV}_5}\}\subset T^*M.
\]
The dual recursion operator \(N^*\), acting componentwise on the pair of one-forms below,
acts on it through a \(2\times2\) control matrix \(M\), defined by
\[
N^*
\begin{pmatrix}
dH_{\mathrm{KdV}_5}\\[2pt]
dK_{\mathrm{KdV}_5}
\end{pmatrix}
=
M
\begin{pmatrix}
dH_{\mathrm{KdV}_5}\\[2pt]
dK_{\mathrm{KdV}_5}
\end{pmatrix}.
\]
A direct computation gives
\[
\begin{cases}
N^*dH_{\mathrm{KdV}_5}
=
(2x-\delta)\,dH_{\mathrm{KdV}_5}
+\dfrac{1}{4a}\,dK_{\mathrm{KdV}_5},
\\[10pt]
N^*dK_{\mathrm{KdV}_5}
=
4a\,y^2\,dH_{\mathrm{KdV}_5}.
\end{cases}
\]
Hence
\[
M=
\begin{pmatrix}
2x-\delta & \dfrac{1}{4a}\\[8pt]
4a\,y^2 & 0
\end{pmatrix}
=
\begin{pmatrix}
u_1+u_2 & \dfrac{1}{4a}\\[8pt]
-4a\,u_1u_2 & 0
\end{pmatrix}.
\]

For the explicit computations it is convenient to multiply \(M\) by the constant factor \(4a\), and to introduce the equivalent matrix
\[
\widetilde M=4a\,M=
\begin{pmatrix}
8ax+2\omega_1-8\omega_2 & 1\\[6pt]
16a^2y^2 & 0
\end{pmatrix}.
\]
The associated compatibility relations are satisfied, and the coefficients
\[
s_1=8ax+2\omega_1-8\omega_2,
\qquad
s_2=16a^2y^2
\]
are functionally independent and in involution. Therefore the eigenvalues of \(\widetilde M\) provide another admissible pair of separating coordinates, namely
\[
\{\widehat Q_1,\widehat Q_2\}=\{4a\,u_1,4a\,u_2\},
\]
up to permutation. Explicitly,
\[
\begin{aligned}
\widehat Q_1&=4ax+\omega_1-4\omega_2
+\sqrt{16a^2x^2+16a^2y^2+8ax\omega_1-32ax\omega_2+\omega_1^2-8\omega_1\omega_2+16\omega_2^2},
\\[4pt]
\widehat Q_2&=4ax+\omega_1-4\omega_2
-\sqrt{16a^2x^2+16a^2y^2+8ax\omega_1-32ax\omega_2+\omega_1^2-8\omega_1\omega_2+16\omega_2^2}.
\end{aligned}
\]
In these variables one has
\[
\widehat Q_1+\widehat Q_2=8ax+2\omega_1-8\omega_2,
\qquad
\widehat Q_1\widehat Q_2=-16a^2y^2.
\]

Proceeding exactly as in the direct computation, one then determines conjugate momenta \(\widehat P_1,\widehat P_2\) such that
\[
\{\widehat Q_1,\widehat P_1\}_0=\{\widehat Q_2,\widehat P_2\}_0=1,
\qquad
\{\widehat Q_1,\widehat Q_2\}_0=\{\widehat P_1,\widehat P_2\}_0=\{\widehat Q_1,\widehat P_2\}_0=\{\widehat Q_2,\widehat P_1\}_0=0.
\]
The inverse canonical transformation is
\[
x=\dfrac{\widehat Q_1+\widehat Q_2-2\omega_1+8\omega_2}{8a},
\qquad
y=\dfrac{\sqrt{-\widehat Q_1\widehat Q_2}}{4a},
\]
\[
p_x=8a\,\dfrac{\widehat Q_1\widehat P_1-\widehat Q_2\widehat P_2}{\widehat Q_1-\widehat Q_2},
\qquad
p_y=8a\sqrt{-\widehat Q_1\widehat Q_2}\,\dfrac{\widehat P_1-\widehat P_2}{\widehat Q_1-\widehat Q_2}.
\]
We now define the rescaled canonical variables unambiguously by
\[
\widehat Q_i=4Q_i,\qquad \widehat P_i=P_i/4,\qquad i=1,2,
\]
after which the canonical transformation takes the simpler form
\[
x=\dfrac{2Q_1+2Q_2-\omega_1+4\omega_2}{4a},
\qquad
y=\dfrac{\sqrt{-Q_1Q_2}}{a},
\]
\[
p_x=2a\,\dfrac{Q_1P_1-Q_2P_2}{Q_1-Q_2},
\qquad
p_y=2a\sqrt{-Q_1Q_2}\,\dfrac{P_1-P_2}{Q_1-Q_2}.
\]

The displayed real chart has \(Q_1Q_2<0\), \(Q_1\ne Q_2\), and chooses \(y=\sqrt{-Q_1Q_2}/a\). The other component is obtained by simultaneously reversing the signs of \(y,p_y\). The Liouville identity \(p_xdx+p_ydy=P_1dQ_1+P_2dQ_2\) verifies canonicity. The coordinates \(Q_i\) are the eigenvalues of \(aN\); equivalently, \(N^*dQ_i=(Q_i/a)dQ_i\) and \(N^*dP_i=(Q_i/a)dP_i\). Thus constant rescaling preserves the eigenspaces and separation.

In these rescaled variables the two integrals acquire separated expressions. More precisely, setting \(H_{\mathrm{KdV}_5}=E\) and \(K_{\mathrm{KdV}_5}=K\), one introduces the function
\[
\Phi(\xi,\pi)
=
\dfrac{
4\xi^5
+(-4\omega_1+24\omega_2)\xi^4
+(\omega_1-4\omega_2)(\omega_1-12\omega_2)\xi^3
+\bigl(32a^4\pi^2+2\omega_1^2\omega_2-16\omega_1\omega_2^2+32\omega_2^3\bigr)\xi^2
+8\mu a^4
}{16\xi a^2}.
\]
 Then one finds
\[
H_{\mathrm{KdV}_5}
=
\dfrac{\Phi(Q_1,P_1)-\Phi(Q_2,P_2)}{Q_1-Q_2},
\]
and, moreover,
\[
\Phi(Q_1,P_1)-E\,Q_1-\dfrac{K}{4}
=
\Phi(Q_2,P_2)-E\,Q_2-\dfrac{K}{4}
=
0.
\]
Thus the separated relations are
\[
\Phi(Q_1,P_1)-E\,Q_1-\dfrac{K}{4}=0,
\qquad
\Phi(Q_2,P_2)-E\,Q_2-\dfrac{K}{4}=0.
\]

Finally, using the canonical equations
\[
\dot Q_1=\dfrac{\partial H_{\mathrm{KdV}_5}}{\partial P_1}
=\dfrac{4a^2P_1Q_1}{Q_1-Q_2},
\qquad
\dot Q_2=\dfrac{\partial H_{\mathrm{KdV}_5}}{\partial P_2}
=\dfrac{4a^2P_2Q_2}{Q_2-Q_1},
\]
one eliminates the momenta and obtains the hyperelliptic form
\[
2(Q_1-Q_2)\dot Q_1=W_1,\qquad 2(Q_2-Q_1)\dot Q_2=W_2,\qquad
W_i=8a^2Q_iP_i,\quad W_i^2=\mathcal P(Q_i),
\]
where
\[
\mathcal P(\xi)
=
-8\xi^5
+(8\omega_1-48\omega_2)\xi^4
-2(\omega_1-4\omega_2)(\omega_1-12\omega_2)\xi^3
+\bigl(32Ea^2-4\omega_1^2\omega_2+32\omega_1\omega_2^2-64\omega_2^3\bigr)\xi^2
+8Ka^2\xi
-16\mu a^4.
\]
Therefore the two Abelian relations, including the original time, are
\[
\frac{dQ_1}{W_1}+\frac{dQ_2}{W_2}=0,\qquad
\frac{Q_1dQ_1}{W_1}+\frac{Q_2dQ_2}{W_2}=\frac{dt}{2}.
\]
Square roots in the examples below denote local signed branches consistent with \(W_i\), rather than simultaneous principal roots. The generic degree-five square-free curve has genus two; repeated roots give degenerate levels.
\begin{remark}
     Unlike the polynomial Kaup--Kupershmidt case, in the KdV$_5$ case the Hamiltonian does not split into a sum of two one-degree-of-freedom Hamiltonians in the separated variables. Rather, the pair \(\bigl(H_{\mathrm{KdV}_5},K_{\mathrm{KdV}_5}\bigr)\) takes a generalized Stäckel form:
\[
H_{\mathrm{KdV}_5}
=
\dfrac{\Phi(Q_1,P_1)-\Phi(Q_2,P_2)}{Q_1-Q_2},
\qquad
K_{\mathrm{KdV}_5}
=
\dfrac{4Q_1\,\Phi(Q_2,P_2)-4Q_2\,\Phi(Q_1,P_1)}{Q_1-Q_2}.
\]
Therefore the separation occurs not at the level of an additive decomposition of the Hamiltonian, but at the level of the common level sets
\[
H_{\mathrm{KdV}_5}=E,\qquad K_{\mathrm{KdV}_5}=K,
\]
which reduce to the two independent separated relations
\[
\Phi(Q_1,P_1)-EQ_1-\dfrac K4=0,
\qquad
\Phi(Q_2,P_2)-EQ_2-\dfrac K4=0.
\]
In this sense the KdV$_5$ system is separated in the Stäckel--Hamilton--Jacobi sense, rather than by a direct splitting of the Hamiltonian vector field into two uncoupled one-degree-of-freedom subsystems.
\end{remark}
\begin{example}
A first natural specialization is obtained by setting
\[
\mu=0.
\]
In this case the Hamiltonian becomes
\[
H_{\mathrm{KdV}_5}
=
\dfrac12\bigl(p_x^2+p_y^2\bigr)
+\dfrac12\bigl(\omega_1x^2+\omega_2y^2\bigr)
+a\,xy^2+2a\,x^3,
\]
while a convenient representative of the second first integral is
\[
\begin{aligned}
K_{\mathrm{KdV}_5}
={}&
4ay\,p_xp_y
+\bigl(4\omega_2-\omega_1-4ax\bigr)p_y^2
\\
&\quad
+a^2y^4
+4a^2x^2y^2
+4a\omega_2xy^2
+\omega_2(4\omega_2-\omega_1)y^2.
\end{aligned}
\]
The separated relations keep the same form,
\[
\Phi(Q_1,P_1)-EQ_1-\dfrac K4=0,
\qquad
\Phi(Q_2,P_2)-EQ_2-\dfrac K4=0,
\]
but now the separation polynomial simplifies to
\[
\mathcal P(\xi)
=
-8\xi^5
+(8\omega_1-48\omega_2)\xi^4
-2(\omega_1-4\omega_2)(\omega_1-12\omega_2)\xi^3
+\bigl(32Ea^2-4\omega_1^2\omega_2+32\omega_1\omega_2^2-64\omega_2^3\bigr)\xi^2
+8Ka^2\xi.
\]
Hence
\[
\mathcal P(\xi)=\xi\,\mathcal Q(\xi),
\]
with \(\mathcal Q\) a quartic polynomial. Therefore the separated equations become
\[
2(Q_1-Q_2)\dot Q_1=\pm\sqrt{Q_1\,\mathcal Q(Q_1)},
\qquad
2(Q_2-Q_1)\dot Q_2=\pm\sqrt{Q_2\,\mathcal Q(Q_2)}.
\]
This example shows that removing the inverse-square term simplifies the hyperelliptic curve, although the Hamiltonian remains of generalized Stäckel type rather than additively separated.
\end{example}

\begin{example}
An even simpler algebraic model is obtained by imposing
\[
\mu=0,
\qquad
\omega_1=\omega_2=0.
\]
Then
\[
H_{\mathrm{KdV}_5}
=
\dfrac12\bigl(p_x^2+p_y^2\bigr)
+a\,xy^2+2a\,x^3,
\]
and
\[
K_{\mathrm{KdV}_5}
=
4ay\,p_xp_y-4ax\,p_y^2+a^2y^4+4a^2x^2y^2.
\]
In separated coordinates \((P_1,P_2,Q_1,Q_2)\), the auxiliary function reduces to
\[
\Phi(\xi,\pi)
=
\dfrac{4\xi^5+32a^4\pi^2\xi^2}{16\xi a^2}
=
\dfrac{\xi^4}{4a^2}+2a^2\pi^2\xi,
\]
so that
\[
H_{\mathrm{KdV}_5}
=
\dfrac{\Phi(Q_1,P_1)-\Phi(Q_2,P_2)}{Q_1-Q_2},
\qquad
K_{\mathrm{KdV}_5}
=
\dfrac{4Q_1\,\Phi(Q_2,P_2)-4Q_2\,\Phi(Q_1,P_1)}{Q_1-Q_2}.
\]
The separated relations are therefore
\[
\dfrac{Q_1^4}{4a^2}+2a^2P_1^2Q_1-EQ_1-\dfrac K4=0,
\qquad
\dfrac{Q_2^4}{4a^2}+2a^2P_2^2Q_2-EQ_2-\dfrac K4=0.
\]
Equivalently,
\[
P_1^2=\dfrac{-Q_1^3+4Ea^2+\dfrac{Ka^2}{Q_1}}{8a^4},
\qquad
P_2^2=\dfrac{-Q_2^3+4Ea^2+\dfrac{Ka^2}{Q_2}}{8a^4},
\]
or, after clearing denominators,
\[
\mathcal P(\xi)=8\xi\bigl(-\xi^4+4Ea^2\xi+Ka^2\bigr).
\]
Thus the two reduced equations of motion become
\[
2(Q_1-Q_2)\dot Q_1=\pm\sqrt{8Q_1\bigl(-Q_1^4+4Ea^2Q_1+Ka^2\bigr)},
\]
\[
2(Q_2-Q_1)\dot Q_2=\pm\sqrt{8Q_2\bigl(-Q_2^4+4Ea^2Q_2+Ka^2\bigr)}.
\]
This is perhaps the simplest nontrivial representative of the KdV$_5$ family: the system is still not additively separated at the Hamiltonian level, but the separated curve and the quadratures take a particularly transparent form.
\end{example}
\begin{example}
In the specialization \(\mu=0\), it is convenient to normalize \(a=1\). If one further chooses
\[
E=0,
\qquad
K=0,
\]
then the separation polynomial becomes
\[
\mathcal P(\xi)
=
-8\xi^5
+(8\omega_1-48\omega_2)\xi^4
-2(\omega_1-4\omega_2)(\omega_1-12\omega_2)\xi^3
-4\omega_2(\omega_1-4\omega_2)^2\xi^2.
\]
Hence
\[
\mathcal P(\xi)=\xi^2\,\widetilde{\mathcal P}(\xi),
\]
so the polynomial has a double zero at \(\xi=0\). An even simpler subcase is obtained by imposing
\[
\omega_1=4\omega_2,
\]
for which
\[
\mathcal P(\xi)=-8\xi^5-16\omega_2\xi^4
=-8\xi^4(\xi+2\omega_2).
\]
The reduced equations then read
\[
2(Q_1-Q_2)\dot Q_1=\pm\sqrt{-8Q_1^4(Q_1+2\omega_2)},
\qquad
2(Q_2-Q_1)\dot Q_2=\pm\sqrt{-8Q_2^4(Q_2+2\omega_2)}.
\]
This provides a particularly transparent example in which the inverse-square term is absent and the separated curve factors nontrivially.
\end{example}

\begin{example}
In the simpler algebraic model
\[
\mu=0,
\qquad
\omega_1=\omega_2=0,
\]
we may also set
\[
a=1.
\]
Then
\[
H_{\mathrm{KdV}_5}
=
\dfrac12\bigl(p_x^2+p_y^2\bigr)+xy^2+2x^3,
\]
and the separation polynomial is
\[
\mathcal P(\xi)=8\xi(-\xi^4+4E\xi+K).
\]
A particularly simple choice is
\[
E=0,
\qquad
K=1,
\]
for which
\[
\mathcal P(\xi)=8\xi(1-\xi^4).
\]
The separated equations become
\[
2(Q_1-Q_2)\dot Q_1=\pm\sqrt{8Q_1(1-Q_1^4)},
\qquad
2(Q_2-Q_1)\dot Q_2=\pm\sqrt{8Q_2(1-Q_2^4)}.
\]
Equivalently,
\[
\dfrac{dQ_1}{\sigma_1\sqrt{8Q_1(1-Q_1^4)}}
=
-\dfrac{dQ_2}{\sigma_2\sqrt{8Q_2(1-Q_2^4)}},\qquad \sigma_i\in\{1,-1\}.
\]
This is perhaps the simplest explicit nontrivial example in the KdV$_5$ family.

Another natural choice is
\[
E=1,
\qquad
K=0,
\]
which yields
\[
\mathcal P(\xi)=8\xi^2(4-\xi^3),
\]
and therefore
\[
2(Q_1-Q_2)\dot Q_1=\pm\sqrt{8Q_1^2(4-Q_1^3)},
\qquad
2(Q_2-Q_1)\dot Q_2=\pm\sqrt{8Q_2^2(4-Q_2^3)}.
\]
In this case the polynomial has a double root at the origin, which makes the factorization especially clear.
\end{example}

\subsection{Case (i): the Sawada--Kotera system}

We consider the generalized Sawada--Kotera Hamiltonian
\[
H_{\mathrm{SK}}
=
\dfrac12\bigl(p_x^2+p_y^2\bigr)
+\dfrac{\omega}{2}(x^2+y^2)
+a\,x y^2+\dfrac{a}{3}x^3+\dfrac{\mu}{2y^2}.
\]
Let us introduce the quantity of total polynomial degree three
\[
K_{\mathrm{SK},0}
=
3p_xp_y+a\,y^3+3a\,x^2y+3\omega xy.
\]
For \(\mu=0\), this is the usual polynomial Sawada--Kotera integral. In the generalized case \(\mu\neq0\), however, it is no longer conserved, since
\[
\{H_{\mathrm{SK}},K_{\mathrm{SK},0}\}_0=-\dfrac{3\mu p_x}{y^3}.
\]
The additional first integral is instead the quartic quantity
\[
K_{\mathrm{SK}}
=
K_{\mathrm{SK},0}^{\,2}
+\mu\left(\dfrac{9p_x^2}{y^2}+12ax\right),
\]
that is,
\[
K_{\mathrm{SK}}
=
\left(3p_xp_y+a\,y^3+3a\,x^2y+3\omega xy\right)^2
+\mu\left(\dfrac{9p_x^2}{y^2}+12ax\right).
\]
Thus the generalized Sawada--Kotera system admits the two first integrals 
\begin{multline}
H_{\mathrm{SK}}
=
\dfrac12\bigl(p_x^2+p_y^2\bigr)
+\dfrac{\omega}{2}(x^2+y^2)
+a\,x y^2+\dfrac{a}{3}x^3+\dfrac{\mu}{2y^2}\\
K_{\mathrm{SK}}
=
\left(3p_xp_y+a\,y^3+3a\,x^2y+3\omega xy\right)^2
+\mu\left(\dfrac{9p_x^2}{y^2}+12ax\right),
\end{multline}
and
\[
\{H_{\mathrm{SK}},K_{\mathrm{SK}}\}_0=0.
\]

\subsubsection{The Sawada--Kotera case with \texorpdfstring{\(\mu=0\)}{mu=0}}

We now treat in detail the polynomial Sawada--Kotera case \(\mu=0\). The Hamiltonian is
\[
H_{\mathrm{SK}}
=
\dfrac12\bigl(p_x^2+p_y^2\bigr)
+\dfrac{\omega}{2}(x^2+y^2)
+a\,x y^2+\dfrac{a}{3}x^3.
\]
A second independent first integral is
\[
F_{\mathrm{SK}}
=
p_xp_y+\omega xy+a\,x^2y+\dfrac{a}{3}y^3.
\]
Equivalently, one may use the quantity of total polynomial degree three
\[
K_{\mathrm{SK},0}=3F_{\mathrm{SK}}
=
3p_xp_y+3\omega xy+3a\,x^2y+a\,y^3.
\]
Hence \((H_{\mathrm{SK}},F_{\mathrm{SK}})\) is a Liouville pair.

\medskip
\noindent
\textbf{The deformation vector field.}
In the coordinate ordering \((p_x,p_y,x,y)\), let
\[
X=
\begin{pmatrix}
-\dfrac{x\,p_x+y\,p_y}{\sqrt2}\\[6pt]
-\dfrac{y\,p_x+x\,p_y}{\sqrt2}\\[6pt]
0\\[2pt]
0
\end{pmatrix}.
\]
Let $P_0$ be the canonical Poisson tensor. Since \(P_0\) is constant, \(P_1=\mathcal L_XP_0\), and one gets
\[
P_1
=-
\begin{pmatrix}
0&0&\dfrac{x}{\sqrt2}&\dfrac{y}{\sqrt2}\\[8pt]
0&0&\dfrac{y}{\sqrt2}&\dfrac{x}{\sqrt2}\\[8pt]
-\dfrac{x}{\sqrt2}&-\dfrac{y}{\sqrt2}&0&0\\[8pt]
-\dfrac{y}{\sqrt2}&-\dfrac{x}{\sqrt2}&0&0
\end{pmatrix}.
\]

\medskip
\noindent
\textbf{The recursion operator.}
Since \(P_0^{-1}=-P_0\), the recursion operator is
\[
N=P_1P_0^{-1}
=
\begin{pmatrix}
\dfrac{x}{\sqrt2}&\dfrac{y}{\sqrt2}&0&0\\[8pt]
\dfrac{y}{\sqrt2}&\dfrac{x}{\sqrt2}&0&0\\[8pt]
0&0&\dfrac{x}{\sqrt2}&\dfrac{y}{\sqrt2}\\[8pt]
0&0&\dfrac{y}{\sqrt2}&\dfrac{x}{\sqrt2}
\end{pmatrix}.
\]
Its characteristic polynomial is
\[
\chi_N(\lambda)
=
\left(\lambda-\dfrac{x+y}{\sqrt2}\right)^2
\left(\lambda-\dfrac{x-y}{\sqrt2}\right)^2,
\]
hence the two eigenvalues, distinct for \(y\ne0\) and each of multiplicity two, are
\[
u_1=\dfrac{x+y}{\sqrt2},
\qquad
u_2=\dfrac{x-y}{\sqrt2}.
\]

\medskip
\noindent
\textbf{The control matrix.}
For the codistribution
\[
\mathcal C=\operatorname{span}\{dH_{\mathrm{SK}},dF_{\mathrm{SK}}\},
\]
one finds
\[
N^*dH_{\mathrm{SK}}
=
\dfrac{x}{\sqrt2}\,dH_{\mathrm{SK}}
+
\dfrac{y}{\sqrt2}\,dF_{\mathrm{SK}},
\qquad
N^*dF_{\mathrm{SK}}
=
\dfrac{y}{\sqrt2}\,dH_{\mathrm{SK}}
+
\dfrac{x}{\sqrt2}\,dF_{\mathrm{SK}}.
\]
Therefore the control matrix is
\[
M=
\begin{pmatrix}
\dfrac{x}{\sqrt2}&\dfrac{y}{\sqrt2}\\[8pt]
\dfrac{y}{\sqrt2}&\dfrac{x}{\sqrt2}
\end{pmatrix},
\]
with characteristic polynomial
\[
\chi_M(\lambda)
=
\left(\lambda-\dfrac{x+y}{\sqrt2}\right)
\left(\lambda-\dfrac{x-y}{\sqrt2}\right).
\]
Thus the eigenvalues of \(M\) are precisely
\[
u_1=\dfrac{x+y}{\sqrt2},
\qquad
u_2=\dfrac{x-y}{\sqrt2},
\]
which are the Darboux--Nijenhuis separating coordinates.

\medskip
\noindent
\textbf{Separating coordinates and conjugate momenta.}
The inverse transformation is
\[
x=\dfrac{u_1+u_2}{\sqrt2},
\qquad
y=\dfrac{u_1-u_2}{\sqrt2},
\]
and the conjugate momenta are
\[
v_1=\dfrac{p_x+p_y}{\sqrt2},
\qquad
v_2=\dfrac{p_x-p_y}{\sqrt2},
\]
that is,
\[
p_x=\dfrac{v_1+v_2}{\sqrt2},
\qquad
p_y=\dfrac{v_1-v_2}{\sqrt2}.
\]
Moreover,
\[
p_x\,dx+p_y\,dy=v_1\,du_1+v_2\,du_2,
\]
so the transformation is globally canonical in the polynomial case. The regular spectral argument uses \(y\ne0\), but this rotation itself extends across \(y=0\).

\medskip
\noindent
\textbf{The Hamiltonian and the second integral in separated variables.}
Using
\[
p_x^2+p_y^2=v_1^2+v_2^2,
\qquad
x^2+y^2=u_1^2+u_2^2,
\qquad
a\,x y^2+\dfrac{a}{3}x^3=\dfrac{\sqrt2\,a}{3}(u_1^3+u_2^3),
\]
we obtain
\[
H_{\mathrm{SK}}
=
\dfrac12(v_1^2+v_2^2)
+\dfrac{\omega}{2}(u_1^2+u_2^2)
+\dfrac{\sqrt2\,a}{3}(u_1^3+u_2^3).
\]
Likewise,
\[
p_xp_y=\dfrac12(v_1^2-v_2^2),
\qquad
xy=\dfrac12(u_1^2-u_2^2),
\qquad
a\,x^2y+\dfrac{a}{3}y^3=\dfrac{\sqrt2\,a}{3}(u_1^3-u_2^3),
\]
hence
\[
F_{\mathrm{SK}}
=
\dfrac12(v_1^2-v_2^2)
+\dfrac{\omega}{2}(u_1^2-u_2^2)
+\dfrac{\sqrt2\,a}{3}(u_1^3-u_2^3),
\]
and equivalently
\[
K_{\mathrm{SK},0}=3F_{\mathrm{SK}}
=
\dfrac32(v_1^2-v_2^2)
+\dfrac{3\omega}{2}(u_1^2-u_2^2)
+\sqrt2\,a\,(u_1^3-u_2^3).
\]

\medskip
\noindent
\textbf{Additive splitting of the Hamiltonian.}
In the polynomial Sawada--Kotera case the Hamiltonian splits additively as in the polynomial Kaup--Kupershmidt case. Setting
\[
H_1(v_1,u_1)
=
\dfrac12v_1^2+\dfrac{\omega}{2}u_1^2+\dfrac{\sqrt2\,a}{3}u_1^3,
\qquad
H_2(v_2,u_2)
=
\dfrac12v_2^2+\dfrac{\omega}{2}u_2^2+\dfrac{\sqrt2\,a}{3}u_2^3,
\]
we get
\[
H_{\mathrm{SK}}=H_1+H_2,
\qquad
F_{\mathrm{SK}}=H_1-H_2,
\qquad
K_{\mathrm{SK},0}=3(H_1-H_2).
\]

\medskip
\noindent
\textbf{The decoupled Hamiltonian systems.}
The Hamilton equations split into the two independent systems
\[
\dot v_1=-\dfrac{\partial H_1}{\partial u_1},
\qquad
\dot u_1=\dfrac{\partial H_1}{\partial v_1},
\qquad
\dot v_2=-\dfrac{\partial H_2}{\partial u_2},
\qquad
\dot u_2=\dfrac{\partial H_2}{\partial v_2},
\]
that is,
\[
\dot v_1=-\omega u_1-\sqrt2\,a\,u_1^2,
\qquad
\dot u_1=v_1,
\]
\[
\dot v_2=-\omega u_2-\sqrt2\,a\,u_2^2,
\qquad
\dot u_2=v_2.
\]
Equivalently,
\[
\ddot u_1+\omega u_1+\sqrt2\,a\,u_1^2=0,
\qquad
\ddot u_2+\omega u_2+\sqrt2\,a\,u_2^2=0.
\]

\medskip
\noindent
\textbf{Quadratures.}
Fixing the values
\[
H_1=E,
\qquad
H_2=K,
\]
one obtains
\[
v_1^2=2 E-\omega u_1^2-\dfrac{2\sqrt2\,a}{3}u_1^3,
\qquad
v_2^2=2 K-\omega u_2^2-\dfrac{2\sqrt2\,a}{3}u_2^3.
\]
Since \(\dot u_i=v_i\), the motion is given by the independent quadratures
\[
\int_{u_{1,0}}^{u_1}\dfrac{d\xi}{\sigma_1\sqrt{\,2 E-\omega \xi^2-\dfrac{2\sqrt2\,a}{3}\xi^3\,}}
=
t-t_0,
\]
\[
\int_{u_{2,0}}^{u_2}\dfrac{d\xi}{\sigma_2\sqrt{\,2 K-\omega \xi^2-\dfrac{2\sqrt2\,a}{3}\xi^3\,}}
=
t-t_0.
\]
Here \(\sigma_i\in\{1,-1\}\) is the sign of \(v_i\) on a local nonturning arc, with continuation on the energy curve at a turning point. The values \(u_{1,0}\), \(u_{2,0}\), and \(t_0\) are arbitrary reference values; equivalently, one may write indefinite integrals and absorb the lower limits into additive integration constants. In terms of the original integrals,
\[
E =\dfrac{H_{\mathrm{SK}}+F_{\mathrm{SK}}}{2}
=\dfrac12H_{\mathrm{SK}}+\dfrac16K_{\mathrm{SK},0},
\qquad
K =\dfrac{H_{\mathrm{SK}}-F_{\mathrm{SK}}}{2}
=\dfrac12H_{\mathrm{SK}}-\dfrac16K_{\mathrm{SK},0}.
\]
\begin{remark}
For the generalized Sawada--Kotera Hamiltonian
\[
H_{\mathrm{SK}}
=
\dfrac12\bigl(p_x^2+p_y^2\bigr)
+\dfrac{\omega}{2}(x^2+y^2)
+a\,x y^2+\dfrac{a}{3}x^3+\dfrac{\mu}{2y^2},
\qquad \mu\neq0,
\]
the polynomial bi-Hamiltonian construction does not extend verbatim. Indeed, the quantity of total polynomial degree three
\[
K_{\mathrm{SK},0}=3p_xp_y+3\omega xy+3a\,x^2y+a\,y^3
\]
is no longer a first integral, since
\[
\{H_{\mathrm{SK}},K_{\mathrm{SK},0}\}_0=-\dfrac{3\mu p_x}{y^3},
\]
and the true additional invariant is the quartic quantity
\[
K_{\mathrm{SK}}
=
K_{\mathrm{SK},0}^{\,2}
+\mu\left(\dfrac{9p_x^2}{y^2}+12ax\right).
\]
Thus the generalized codistribution is generated by a Hamiltonian quadratic in the momenta and an integral quartic in the momenta. For \(\mu=0\), both \(H_{\mathrm{SK}}\) and the unsquared \(F_{\mathrm{SK}}\) are quadratic in the momenta.

Moreover, the singular term \(\mu/(2y^2)\) introduces rational expressions into the Hamilton equations and into the invariance conditions for the codistribution, so a rational ansatz for the deformation vector field \(X\) is natural, although this observation does not prove that every possible ansatz must be rational. At the same time, the rotated variables
\[
u_1=\dfrac{x+y}{\sqrt2},
\qquad
u_2=\dfrac{x-y}{\sqrt2},
\qquad
v_1=\dfrac{p_x+p_y}{\sqrt2},
\qquad
v_2=\dfrac{p_x-p_y}{\sqrt2},
\]
no longer separate the Hamiltonian additively, since
\[
H_{\mathrm{SK}}
=
\dfrac12(v_1^2+v_2^2)
+\dfrac{\omega}{2}(u_1^2+u_2^2)
+\dfrac{\sqrt2\,a}{3}(u_1^3+u_2^3)
+\dfrac{\mu}{(u_1-u_2)^2}.
\]
Therefore, although one may still search for a second Poisson tensor of the form \(P_1=\mathcal L_XP_0\), the determining system for \(X\) becomes much more involved than in the case \(\mu=0\). We do not construct such a generalized SK tensor here. Canonical transformations relating SK and KK, and spectral-curve methods, provide the generalized separation; see \cite{BRW94,VMC02,EFH15}.
\end{remark}

% ============================================================
% Declarations (place this block just before the References)
% ============================================================

\section*{Declarations}

\subsection*{Funding}
Alessandro Portaluri was partially supported by the Tamkeen Portaluri 2025--26 Faculty Research Funds, Tamkeen, NYU Abu Dhabi (Program Code: ADHPG; Project Code: VP255). Nicola Sottocornola received no funds, grants, or other support.

\subsection*{Competing interests}
The authors declare that they have no competing interests relevant to the content of this article.

\subsection*{Data availability}
No datasets were generated or analysed during the current study.

\subsection*{Code availability}

The Maple worksheets and scripts used for the original computations are available from the corresponding author upon reasonable request. Independent exact symbolic verification scripts are provided with the revised source files.

\vskip1cm

\noindent
\textsc{Alessandro Portaluri}\\
Università degli Studi di Torino (DISAFA)\\
Largo Paolo Braccini 2\\
10095 Grugliasco (Torino), Italy\\
\texttt{alessandro.portaluri@unito.it}

\medskip
\noindent
Visiting Professor of Mathematics\\
New York University Abu Dhabi\\
Saadiyat Marina District, Abu Dhabi, UAE\\
\texttt{ap9453@nyu.edu}

\vskip0.6cm

\noindent
\textsc{Nicola Sottocornola}\\
Lecturer of Mathematics\\
New York University Abu Dhabi\\
Saadiyat Marina District, Abu Dhabi, UAE\\
\texttt{ns6159@nyu.edu}

\end{document}